\documentclass[aos]{imsart}

\RequirePackage{amsthm,amsmath,amsfonts,amssymb}
\RequirePackage[numbers,sort&compress]{natbib}
\RequirePackage[colorlinks,citecolor=blue,urlcolor=blue]{hyperref}
\RequirePackage{graphicx, tikz, subcaption, dsfont}
\usetikzlibrary{graphs, positioning, calc}

\startlocaldefs
\theoremstyle{plain}

\newtheorem{theorem}{Theorem}[section]

\newtheorem{proposition}{Proposition}

\theoremstyle{remark}

\newtheorem*{example}{Example}

\newtheorem{assumption}{Assumption}

\endlocaldefs

\usepackage{xcolor}

\usepackage[normalem]{ulem}
\newcommand{\stkout}[1]{\ifmmode\text{\sout{\ensuremath{#1}}}\else\sout{#1}\fi}

\begin{document}

\begin{frontmatter}
\title{Proximal Causal Inference For Hidden Outcomes}
\runtitle{Proximal Causal Inference for Hidden Outcomes}

\begin{aug}
\author[A]{\fnms{Helen}~\snm{Guo}\ead[label=e1]{hguo51@jh.edu}},
\author[B]{\fnms{AmirEmad}~\snm{Ghassami}\ead[label=e3]{ghassami@bu.edu}},
\author[C]{\fnms{Ilya}~\snm{Shpitser}\ead[label=e4]{ilyas@cs.jhu.edu}},
\and
\author[A]{\fnms{Elizabeth L.}~\snm{Ogburn}\ead[label=e2]{eogburn@jhu.edu}}

\address[A]{Department of Biostatistics,
Johns Hopkins Bloomberg School of Public Health \printead[presep={,\ }]{e1,e2}}

\address[B]{Department of Mathematics and Statistics,
Boston University\printead[presep={,\ }]{e3}}

\address[C]{Department of Computer Science,
Johns Hopkins Whiting School of Engineering \printead[presep={,\ }]{e4}}

\end{aug}

\begin{abstract}
Methods that rely on proxies, without imposing strong parametric structure, are increasingly used to deal with unobserved variables in causal inference. One influential line of this work reconstructs latent distributions used to identify the target functional by exploiting eigenvalue-eigenvector structure. Within this framework, we first establish identification of the full-data law in the presence of hidden outcomes, and then develop influence function-based estimators for causal effects. To the best of our knowledge, this is the first work to develop influence function-based estimators in this setting without relying on unbiased proxy measurements or partial observation, yielding examples that are multiply robust and for which product-bias conditions for \(\sqrt{n}\)-consistent estimation can be derived. We demonstrate the performance of our approach through simulation studies.
\end{abstract}

\end{frontmatter}

\section{Introduction}\label{sec:1} 
Hidden variables remain a central challenge in causal inference, often preventing identification and estimation of causal targets of interest. Proxy-based approaches, which leverage key observed variables associated with the hidden variable(s), informally referred to as ``proxies,'' have recently emerged as a way to circumvent such hidden structure without relying on strong parametric assumptions \citep{kuroki14measurement, allman2009identifiability, miao2018identifying, ghassami24causal, zhou2024causalinferencehiddentreatment}. A prominent line of this work achieves identification by reconstructing latent distributions from observed data under relatively mild conditions, exploiting eigenvalue–eigenvector structure \citep{kuroki14measurement, allman2009identifiability, zhou2024causalinferencehiddentreatment}.

These methods are often described as “nonparametric,” in the sense that their assumptions are sufficiently weak; under relative support size/dimensionality conditions, distributions that violate them can be approximated arbitrarily closely by distributions that satisfy them \citep{semiparametric2020cui, Canay_testability2013}. Conceptually, this line of work builds on ideas dating back to Kruskal’s uniqueness theorem for the Candecomp/Parafac (CP) decomposition \citep{kruskal1977three}, which characterizes when a three-way array admits a unique representation as a sum of rank-one components. This result allows recovery of the joint distribution of hidden and observed variables (i.e., the full-data law).

Using this framework, we establish identification of the full-data law in the presence of proxies for hidden outcomes and develop an influence function-based approach for estimating causal effects. Our influence function-based approach yields examples of estimators that are \emph{multiply robust} (consistent even if certain nuisance functions are misspecified) and achieve \(\sqrt{n}\)-consistent estimation under product-bias conditions on nuisance estimation errors.

A substantial body of literature has studied effect identification and estimation under missing outcomes, as well as partial identification under various measurement error assumptions \citep{ding2018, chen2023, phung24zero, Imai2010, vanderweele2012results, Duarte2024}. Additionally, recent work has considered proxy-based causal inference methods for hidden outcomes in the presence of at least one unbiased proxy \citep{ Fu2026}. To the best of our knowledge, our work is the first to establish semiparametric estimation results for hidden outcomes that do not rely on unbiased proxies or partial observations of the true outcome. We assess the performance of our proposed estimators through simulation studies.
\newpage
\section{Background}

\subsection{Directed Acyclic Graphs (DAGs)}

We use a causal directed acyclic graph (DAG) as a convenient representation of a setting consistent with our assumptions (Figure~\ref{fig:hidden_outcome}). DAGs encode a statistical model; nodes correspond to random variables, and edges encode conditional independence relations, with all implied independences determined via \emph{d-separation}. Additionally, causal DAGs provide a transparent way to encode assumptions about the direction of causal relationships and the conditions under which causal effects can be \emph{identified}, i.e., expressed as a function of the joint distribution of observed variables (i.e., the observed-data law). \citep{pearl88probabilistic, pearl2009causality, spirtes01causation, richardson13swig}. 

Directed edges connecting \(A\), \(Y\), and \(C\) in Figure~\ref{fig:hidden_outcome} depict the standard confounding structure considered in causal inference problems, with \(C\) acting as a common cause of treatment \(A\) and outcome \(Y\). Dotted edges represent optional causal relationships and may be omitted without affecting any of the results developed in this paper.

\subsection{Confounding Adjustment Formula}

Causal effects are defined through interventional distributions that describe how outcomes $\vec{Y}$ would change if treatment variable(s) $\vec{A}$ were externally set to values $\vec{a}$ via the intervention operator $\text{do}(\vec{a})$ \citep{pearl2009causality}. These interventional quantities can equivalently be written in potential outcomes notation as $p(\vec{Y}(\vec{a}))$. If all variables in Figure~\ref{fig:hidden_outcome} are observed, then under Assumptions~\ref{assump:positivity}--\ref{assump:consistency}, together with Assumption~\ref{assump:ignorability} encoded by the causal DAG, the potential outcome distribution $p(Y(a))$ is identified for every \(a \in \mathcal{A}\) via the confounding adjustment formula in Equation~\ref{eq:adjustment}. For binary \(A\), this yields the causal effect
\(
\mathbb{E}[Y(1)] - \mathbb{E}[Y(0)].
\)

\begin{equation}\label{eq:adjustment}
p(Y(a)) \;=\; \int p(Y \mid A = a, C = c)\, p(C = c)\, dc.
\end{equation}

\begin{assumption}[Positivity]\label{assump:positivity} 
For each \(a \in \mathcal{A}\),
\[p(A=a \mid C) > 0.\]
\end{assumption}

\begin{assumption}[Consistency]\label{assump:consistency}
For each \(a \in \mathcal{A}\),
\[Y = Y(a) \text{ when } A = a.\]
\end{assumption}

\begin{assumption}[Conditional Ignorability] \label{assump:ignorability} For each \(a \in \mathcal{A}\),
\[
Y(a) \perp\!\!\!\perp A \mid C.
\]
\end{assumption}

\subsection{Related Work}\label{sec:related_work}

Our work builds on classic results such as Kruskal’s uniqueness theorem for the Candecomp/Parafac (CP) decomposition \citep{kruskal1977three} and subsequent developments \citep{kuroki14measurement, allman2009identifiability}, which represent the observed-data law as a multi-way array. Under conditional independence of the observed variables given latent variable(s), together with appropriate rank conditions, this array admits a unique factorization that recovers the full-data law. The underlying argument hinges on unique eigenspaces \citep{rhodes20101818}; obtaining the full-data law from the observed-data law can often be achieved via a sequence of eigendecomposition steps \citep{kuroki14measurement, allman2009identifiability}, or alternatively through iterative search procedures (e.g., alternating least squares, gradient descent) \citep{ALS, gradient_descent}.

Continuous analogues replace these finite-dimensional decompositions with operator-based formulations, relying on \emph{completeness} conditions characterizing the variability of conditional distributions. This perspective was developed by \citet{hu2008instrumental} and extended to causal inference settings in the presence of proxies for hidden confounders by \citet{deaner2023controllinglatentconfoundingtriple}, as well as to settings with proxies for hidden treatments by \citet{zhou2024causalinferencehiddentreatment}. Obtaining the full-data law from the observed-data law can be achieved through sieve-based or expectation-maximization (EM) approaches \citep{hu2008instrumental,zhou2024causalinferencehiddentreatment}.

We build on the ideas from \citet{hu2008instrumental}, \citet{deaner2023controllinglatentconfoundingtriple}, and 
\citet{zhou2024causalinferencehiddentreatment} by adapting such identifying assumptions to settings with analogous proxies for hidden outcomes. In addition, we develop influence function--based estimators for the resulting causal effects, yielding examples that are multiply robust and \(\sqrt{n}\)-consistent under product-bias conditions on nuisance estimation errors.

\section{Full-data Law Identification}\label{sec:identification}

We leverage the array decomposition framework discussed in Section \ref{sec:related_work} to identify
causal effects in settings where hidden outcomes have (possibly biased) proxies. We then develop novel influence function--based estimators for identified effects that are provably multiply robust and \(\sqrt{n}\)-consistent under product-bias conditions.

Consider a binary treatment \(A\), hidden outcome \(Y\), and observed confounders \(C\). Even under Assumptions~\ref{assump:positivity}--\ref{assump:ignorability}, causal effects are not generally identifiable since \(Y\) is unobserved. Now suppose that additional variables \(W\), \(Z\), and \(V\) (referred to henceforth as ``proxies'' for the hidden outcome) are observed, and that Assumptions~\ref{assump:boundedness_hidden_Y}--\ref{assump:unique_mapping} hold. Then identification of the full-data law, and thereby causal effects, becomes possible. A directed acyclic graph consistent with these assumptions is shown in Figure~\ref{fig:hidden_outcome}.
\begin{figure}[h!]
    \centering
        \centering
        \begin{tikzpicture}[>=latex, node distance=1.4cm, anchor=center]
            \node (A) [draw, circle] {A};
            \node (Y) [draw, circle, fill=red!30, right of=A] {Y};
            \node (C) [draw, circle, above of=A, xshift=0.7cm] {C};
            \node (W) [draw, circle, below of=Y, xshift=-0.7cm] {W};
            \node (Z) [draw, circle, below of=Y, xshift=0.6cm] {Z};
            \node (V) [draw, circle, below of=Y, xshift=1.7cm] {V};

            \draw[->] (A) -- (Y);
            \draw[->] (Y) -- (W);
            \draw[->] (Y) -- (Z);
            \draw[->] (Y) -- (V);
            \draw[->, dotted] (A) -- (Z);
            \draw[->, dotted] (A) -- (W);
            \draw[->, dotted] (A) -- (V);
            \draw[->] (C) -- (Y);
            \draw[->] (C) -- (A);
            \draw[->, dotted] (C) to[out=10, in=90] (Z);
            \draw[->, dotted] (C) -- (W);
            \draw[->, dotted] (C) to[out=10, in=90] (V);
        \end{tikzpicture}
        \caption{Hidden Outcome Model \\
        Supports \(\mathcal{A},\mathcal{Y}, \mathcal{C}, \mathcal{W},\mathcal{Z},\mathcal{V}\)}
        \label{fig:hidden_outcome}
    \hfill
\end{figure}
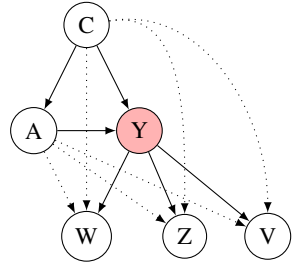

\begin{assumption}\label{assump:boundedness_hidden_Y}
$(A,Y,C,W,Z,V)$ admit a joint density that is bounded.

\end{assumption}

\begin{assumption}\label{assump:ci_hidden_y}
    $W, Z, V$ are mutually independent $\mid Y, A, C$
\end{assumption}

\begin{assumption}\label{assump:completeness_hidden_Y}
For each \((a,c) \in \mathcal{A} \times \mathcal{C}\), for
\(
g \in L^2,
\)
the following completeness conditions hold:
\[
\mathbb{E}[g(Y)\mid W,a,c]=0
\ \text{a.s.}
\iff
g(Y)=0
\ \text{a.s.}, \text{ and}
\]
\[
\mathbb{E}[g(Y)\mid Z,a,c]=0
\ \text{a.s.}
\iff
g(Y)=0
\ \text{a.s.}
\]
\end{assumption}

Assumption~\ref{assump:completeness_hidden_Y}, a set of \emph{completeness}
conditions, states that for each \((a,c) \in \mathcal{A} \times \mathcal{C}\), the conditional distributions
\[
\{p(W \mid y,a,c) : y \in \mathcal{Y}_{a,c}\}
\quad \text{and} \quad
\{p(Z \mid y,a,c) : y \in \mathcal{Y}_{a,c}\}
\]
vary sufficiently across \(y\), 
where \(\mathcal{Y}_{a,c}\) denotes the support of \(Y \mid A=a,C=c\). These conditions can be interpreted as a full-rank requirement on linear operators defined by these conditional distributions. For further discussion of this, see \citet{guo2025}. These completeness assumptions are flexible in the sense that in continuous models, if \(W\) and \(Z\) each have larger dimension than \(Y\) for each \((a,c)\), distributions that violate them can be approximated arbitrarily closely by distributions that satisfy them \citep{semiparametric2020cui, Canay_testability2013}. 

\begin{assumption}\label{assump:distinctness_hidden_Y}
For each \((a,c) \in \mathcal{A} \times \mathcal{C}\), for
\(y_i, y_j \in \mathcal{Y}_{a,c}\) with \(y_i \neq y_j\), the set
\(
\{v : p(v \mid y_i, a, c) \neq p(v \mid y_j, a, c)\}
\)
has positive probability under the marginal of \(V \mid A=a, C=c\).
\end{assumption}

Note that \(W\), \(Z\), and \(V\) play symmetric roles in these assumptions; the role of these variables may be interchanged with one another. Notably, Assumption~\ref{assump:distinctness_hidden_Y} is weaker than a completeness condition. In particular, a completeness statement of the form given in Assumption~\ref{assump:completeness_hidden_Y}, with $V$ in place of $W$ or $Z$, implies Assumption~\ref{assump:distinctness_hidden_Y}.

Assumptions~\ref{assump:boundedness_hidden_Y}--\ref{assump:distinctness_hidden_Y} recover \(p(Y=y_i,W,Z,V \mid A=a,C=c)\) up to the label \(y_i\) for each \((a,c)\). If \(Y\) has finite support, then, for each \((a,c)\), Assumptions~\ref{assump:boundedness_hidden_Y}--\ref{assump:distinctness_hidden_Y} imply a unique support size for \(\mathcal{Y}_{a,c}\) \citep{Derksen2013}. If \(\mathcal{Y}_{a,c}=\mathbb{R}^n\), the support dimension \(n\) is similarly determined under mild regularity conditions \citep{hu2008instrumental}. For practical purposes, however, we assume that the support of \(\mathcal{Y}_{a,c}\) is known for each \((a,c)\), simplifying search procedures for \(p(Y,W,Z,V \mid A=a,C=c)\).

\begin{assumption}\label{assump:known_support}
 For each
\(
(a,c) \in \mathcal{A} \times \mathcal{C}\), the support of \(\mathcal{Y}_{ac}\) is known.
\end{assumption}

Assumption~\ref{assump:unique_mapping} provides two different sufficient conditions for recovering the labels of \(Y\). In condition (i), following \citet{hu2008instrumental}, \(M\) may be any known functional of the conditional distribution in its argument. In condition (ii), a strengthening of Assumption 9 from \citet{deaner2023controllinglatentconfoundingtriple}, we first identify the ordering of the latent states of \(Y\). Then, given known support, the labels of \(Y\) can be recovered.

\begin{assumption}\label{assump:unique_mapping}
At least one of the following holds:
\renewcommand{\theenumi}{\roman{enumi}}%
\begin{enumerate}
\item For each
\(
(a,c) \in \mathcal{A} \times \mathcal{C}\), there exists a known functional \(M\) such that for some \(X \in \{W,Z,V\}\), 
\[
M\!\left[p(X \mid y,a,c)\right] = y
\text{ 
for all }
y \in \mathcal{Y}_{a,c}.
\]
\item $Y$ has finite-support, is ordinal, and for each \((a,c) \in \mathcal{A} \times \mathcal{C}\), for
\(y_i, y_j \in \mathcal{Y}_{a,c}\) with \(y_i \neq y_j\),
there exists a known functional \(M\) such that for some \(X \in \{W,Z,V\}\), 
\[
M\!\left[p(X \mid y_i,a,c)\right]
<
M\!\left[p(X \mid y_j,a,c)\right].
\]
\end{enumerate}
\end{assumption}

\begin{theorem}\label{thm:id}
Under Assumptions~\ref{assump:boundedness_hidden_Y}-\ref{assump:unique_mapping}, the full-data law of hidden and observed variables \(p(A,Y,C,W,Z,V)\) is identified.
\end{theorem}

 A proof of Theorem~\ref{thm:id} is provided in Appendix~\ref{app:identification}, based on Theorem~1 in \citet{hu2008instrumental}.

Since the full-data law is identified, potential outcome distribution \(p(Y(a))\) is identified via Equation~\ref{eq:adjustment} under Assumptions~\ref{assump:positivity}-\ref{assump:ignorability}, yielding identification of the causal effect \(\mathbb{E}[Y(1)] - \mathbb{E}[Y(0)]\). This naturally raises the question of how to construct estimators with desirable properties, such as \(\sqrt{n}\)-consistency and multiple robustness.

\section{Estimation of Causal Effects}\label{sec:estimation}

Influence functions characterize the local sensitivity of a parameter to perturbations of the data-generating distribution, and can yield estimators that are multiply robust (i.e., consistent under partial model misspecification) and \(\sqrt{n}\)-consistent despite nuisance functions being estimated at slower-than-parametric \(\sqrt{n}\) rates, making them particularly useful when nuisance functions are estimated flexibly \citep{KennedyTutorial}.
 
\subsection{Full-data Influence Functions}
The full data consist of \(H = (A,Y,C,W,Z,V)\), where \(Y\) is unobserved. Let \(\psi\) be a functional of the full-data law \(p(H)\), indexed by nuisance functions \(\eta\). For a regular parametric submodel \(\{P_\epsilon(H) : \epsilon \in \mathbb{R}\}\) with corresponding densities \(p_\epsilon(H)\) and score
\[
s(H) = \left.\frac{\partial}{\partial \epsilon} \log p_\epsilon(H)\right|_{\epsilon=0},
\]
a \emph{full-data influence function} \(\phi_{\mathrm{full}}(H; \psi, \eta)\) satisfies
\[
\left.\frac{d}{d\epsilon}\psi(P_\epsilon(H))\right|_{\epsilon=0}
= \mathbb{E}\big[\phi_{\mathrm{full}}(H; \psi, \eta)\, s(H)\big],
\qquad 
\mathbb{E}\big[\phi_{\mathrm{full}}(H; \psi, \eta)\big]=0,
\]
for all such submodels. For example, under Assumptions~\ref{assump:positivity}-\ref{assump:ignorability}, we give a full-data influence function \(\phi_{\mathrm{full}}(H; \psi_a, \eta)\) for \(\psi_a = \mathbb{E}[Y(a)] = \int p(Y \mid A=a, C=c)p(C=c)dc\) in Equation~\ref{eq:full_IF}. A proof is provided in Appendix~\ref{app:full_IF}.
\begin{equation}\label{eq:full_IF}
\phi_{\mathrm{full}}(A,Y,C; \psi_a, \eta) 
= \frac{\mathbb{I}(A=a)}{p(A=a \mid C)}
   \Big\{ Y - \mathbb{E}[Y \mid A=a, C] \Big\} 
   + \mathbb{E}[Y \mid A=a, C] - \psi_a .
\end{equation}
Because \(Y\) is unobserved, \(\phi_{\mathrm{full}}(H; \psi_a, \eta)\) is not directly available.

\subsection{Observed-data Influence Functions}
The observed data consist of \(O = (A,C,W,Z,V)\). Let \(\psi\) be a functional of the observed-data law \(p(O)\), indexed by nuisance functions \(\nu\).  For a regular parametric submodel \(\{P_\epsilon(O): \epsilon \in \mathbb{R}\}\) with corresponding densities \(p_\epsilon(O)\) and score
\[
s(O) = \left.\frac{\partial}{\partial \epsilon} \log p_\epsilon(O)\right|_{\epsilon=0},
\]
an \emph{observed-data influence function} \(\phi_{\mathrm{obs}}(O; \psi, \nu)\) satisfies
\[
\left.\frac{d}{d\epsilon}\psi(P_\epsilon(O))\right|_{\epsilon=0}
= \mathbb{E}\big[\phi_{\mathrm{obs}}(O; \psi, \nu)\, s(O)\big],
\qquad 
\mathbb{E}\big[\phi_{\mathrm{obs}}(O; \psi, \eta)\big]=0,
\]
for all such submodels.

\subsection{Estimation Framework}
Under Assumptions~\ref{assump:positivity}--\ref{assump:unique_mapping}, the average counterfactual outcome \(\psi_a = \mathbb{E}[Y(a)] = \int p(Y \mid A=a, C=c)p(C=c)dc\) is identified, and thereby a functional of the observed-data law. We will construct an observed-data influence function \(\phi_{\mathrm{obs}}(O; \psi_a, \eta)\)\footnote{Note that under full-law identification, the observed-data nuisance functions $\eta$ and the full-data nuisance
functions $\nu$ are in one-to-one correspondence, so that either may be used to index the observed distribution.} for this target.

An estimator \(\hat\psi_a\) is obtained by solving the estimating equation
\begin{equation}\label{eq:estimation}
   \mathbb{P}_n\{\phi_{\mathrm{obs}}(O; \hat\psi_a, \hat\eta)\} = 0,
\end{equation}
where \(\mathbb{P}_n\) denotes the empirical average. Under regularity conditions, \(\hat\psi_a\) is \(\sqrt{n}\)-consistent and asymptotically normal with asymptotic variance
\(
\mathbb{E}\big[\phi_{\mathrm{obs}}^2(O; \psi_a, \eta)\big]
\) \cite{tsiatis06missing}.

The average treatment effect \(\psi = \mathbb{E}[Y(1)] - \mathbb{E}[Y(0)]\) can then be estimated by
\[
\hat\psi = \hat\psi_1 - \hat\psi_0.
\]
This estimator follows directly from the construction of \(\hat\psi_a\), inheriting multiple robustness properties \(\hat\psi_a\) may possess, and is \(\sqrt{n}\)-consistent and asymptotically normal under regularity conditions, with asymptotic variance
\(
\mathbb{E}\Big[\big\{\phi_{\mathrm{obs}}(O; \psi_1, \eta) - \phi_{\mathrm{obs}}(O; \psi_0, \eta)\big\}^2\Big]
\).

\subsection{Influence Function Construction}

We construct an observed-data influence function 
\(\phi_{\mathrm{obs}}(O; \psi_a, \eta)\) for target \(\psi_a\) by replacing the unobserved random variable \(Y\) in the full-data influence function in Equation~\ref{eq:full_IF} with a function \(\Omega(O;\eta)\) of the observed data respecting equations delineated in Assumption~\ref{assump:weight}.

\begin{assumption}\label{assump:weight}
There exists a function \(\Omega(O;\eta)\) such that for each \(X \in \{W,Z,V\}\),
\begin{equation}\label{eq:bridge}
\mathbb{E}\big[\Omega(O;\eta)\mid Y,A,C,X\big] = Y.
\end{equation}
\end{assumption}

\begin{proposition}\label{prop:guaranteed_weight}
Suppose \(Y\) has finite support. Further suppose that we strengthen Assumption~\ref{assump:distinctness_hidden_Y} to a completeness conditions; that is, for each \((a,c) \in \mathcal{A} \times \mathcal{C}\), for \( g \in L^2, \) \[ \mathbb{E}[g(Y)\mid V,a,c]=0 \ \text{a.s.} \iff g(Y)=0 \ \text{a.s.}\]
Then under Assumptions~\ref{assump:boundedness_hidden_Y}-\ref{assump:unique_mapping}, there exists weights satisfying that for every \(X \in \{W,Z,V\}\),
\begin{equation}\label{eq:indicator}
\mathbb{E}\big[\omega_y(O;\eta)\mid Y=y',A,C,X\big]
=
\mathbf{1}\{y=y'\}
\end{equation}
so that \[
\Omega(O;\eta)
=
\sum_{y \in \mathcal{Y}_{A,C}}
\omega_y(O;\eta)\, y,
\]
where \(\mathcal{Y}_{A,C}\) denotes the support of \(Y\mid A,C\), satisfies Assumption~\ref{assump:weight}.
\end{proposition}

A constructive proof of Proposition~\ref{prop:guaranteed_weight} is provided in Appendix~\ref{app:guaranteed_weight}, although the nuisance functions involved are relatively complex. Appendix~\ref{app:weights} provides a procedure for constructing weights \(\omega_y(O;\eta)\) constituting \(\Omega(O;\eta)\) in terms of more tractable nuisance functions. We give an example below. Our strategy begins by selecting a collection of basis functions for each proxy variable. Generally, we may select appropriate bases from a large dictionary of function classes including indicator functions, polynomials, splines, and Fourier features.

\begin{example}\label{ex:binary}
Suppose that \(Y\) is binary conditional on each \((a,c)\), and
\[
(W,Z,V)\mid Y=y_i,A,C
\sim
N\!\left(
\begin{pmatrix}
\mu_{W,i}(A,C)\\
\mu_{Z,i}(A,C)\\
\mu_{V,i}(A,C)
\end{pmatrix},
\begin{pmatrix}
\sigma_W^2(A,C) & 0 & 0\\
0 & \sigma_Z^2(A,C) & 0\\
0 & 0 & \sigma_V^2(A,C)
\end{pmatrix}
\right).
\]
Assume that the conditional means are distinct across
latent states:
\[
\mu^W_i(A,C) \neq \mu^W_j(A,C),\qquad
\mu^Z_i(A,C) \neq \mu^Z_j(A,C),\qquad
\mu^V_i(A,C) \neq \mu^V_j(A,C),
\qquad i\neq j.
\] 

Choose two basis functions
\[
q_{W,1}(W)=1,
\qquad
q_{W,2}(W)=W.
\]
Define the conditional moment matrix
\[
\mathbf M_W(A,C)
=
\begin{pmatrix}
\mathbb{E}[q_{W,1}\mid y_1,A,C] & \mathbb{E}[q_{W,1}\mid y_2,A,C] \\
\mathbb{E}[q_{W,2}\mid y_1,A,C] & \mathbb{E}[q_{W,2}\mid y_2,A,C]  
\end{pmatrix}.
\]

By our chosen basis, 
\[
\mathbf M_W(A,C)
=
\begin{pmatrix}
1 & 1 
\\
\mathbb{E}(W\mid Y=y_1,A,C) &
\mathbb{E}(W\mid Y=y_2,A,C) 
\end{pmatrix}.
\]

Since \(\mathbf M_W(A,C)\) is invertible,
define
\[
h_W(W,A,C;\eta)
=
\mathbf M_W(A,C)^{-1}
\begin{pmatrix}
1\\
W
\end{pmatrix}.
\]
Then for the \(i\)th component of \(h_W(W,A,C;\eta)\), denoted \(h_{W,i}(W,A,C;\eta)\), we have that
\(
\mathbb E\!\left[
h_{W,i}(W,A,C;\eta)
\mid Y=y_j,A=a,C=c
\right]
=
\mathbf 1\{i=j\}.
\)

Similarly, construct \(h_Z(Z,A,C;\eta)\) and \(h_V(V,A,C;\eta)\) from
\(\mathbf M_Z(A,C)\) and \(\mathbf M_V(A,C)\), respectively. Finally, define
\begin{equation}\label{eq:binary_weight1}
\begin{aligned}
\omega_{y_i}(O;\eta)
&=
h_{W,i}(W,A,C;\eta)h_{Z,i}(Z,A,C;\eta)
+
h_{W,i}(W,A,C;\eta)h_{V,i}(V,A,C;\eta) \\
&\quad +
h_{Z,i}(Z,A,C;\eta)h_{V,i}(V,A,C;\eta)
-
2h_{W,i}(W,A,C;\eta)h_{Z,i}(Z,A,C;\eta)h_{V,i}(V,A,C;\eta).
\end{aligned}
\end{equation}

Equivalently,
\begin{equation}\label{eq:binary_weight2}
\begin{aligned}
\omega_{y_i}(O;\eta)
&=
\frac{(W-\mu^W_j)(Z-\mu^Z_j)}
     {(\mu^W_i-\mu^W_j)(\mu^Z_i-\mu^Z_j)}
+
\frac{(Z-\mu^Z_j)(V-\mu^V_j)}
     {(\mu^Z_i-\mu^Z_j)(\mu^V_i-\mu^V_j)} \\
&\quad +
\frac{(W-\mu^W_j)(V-\mu^V_j)}
     {(\mu^W_i-\mu^W_j)(\mu^V_i-\mu^V_j)}
-
2\frac{(W-\mu^W_j)(Z-\mu^Z_j)(V-\mu^V_j)}
       {(\mu^W_i-\mu^W_j)(\mu^Z_i-\mu^Z_j)(\mu^V_i-\mu^V_j)},
\qquad i\neq j,
\end{aligned}
\end{equation}
where
\(
\mu_{X,i}=\mathbb E[X\mid Y=y_i,A,C]\) and \(
\mu_j^X=\mathbb E[X\mid Y=y_j,A,C]
\)
for \(X\in\{W,Z,V\}\).

Then
\[
\Omega(O;\eta)
=
\sum_{y_i \in \mathcal{Y}_{A,C}}
\omega_{y_i}(O;\eta) y_i
\]
satisfies
\[
\mathbb{E}\!\left[
\Omega(O;\eta)
\;\middle|\;
Y,A,C,X
\right]
=
Y,
\qquad
X \in \{W,Z,V\}.
\]
\end{example}

\begin{proposition}\label{prop:polynomial_basis}
Suppose that \(Y\) has finite support, Assumptions~\ref{assump:boundedness_hidden_Y}-\ref{assump:unique_mapping} hold, and 
\[
X \mid Y=y_i,A,C
\sim
\mu_{X,i}(A,C)+\epsilon_X(A,C),
\qquad X\in\{W,Z,V\},
\]
where error term \(\epsilon_X(A,C)\) does not depend on \(Y\) and for \(i \neq j\), \(\mu_{X,i}(A,C) \neq \mu_j^X(A,C)\). Then the weights \(\omega_y(O;\eta)\) in
Proposition~\ref{prop:guaranteed_weight} may be constructed using polynomial
bases. In particular, in Appendix~\ref{app:weights}, one may define
\(
q_X(X,A,C)
=
(1,X,\ldots,X^{R(A,C)-1})^\top
\) for each \(X \in \{W,Z,V\}\). 
\end{proposition}

We prove Proposition~\ref{prop:polynomial_basis} in Appendix~\ref{app:polynomial_basis}.

When \(Y\) is continuous, characterizing existence of \(\Omega(O;\eta)\) in Assumption~\ref{assump:weight} is more challenging.
We omit further details, but note that without an explicit functional representation of the solution \(\Omega(O;\eta)\) in terms of nuisance functions \(\eta\), establishing multiple robustness properties and product-bias conditions sufficient for \(\sqrt{n}\)-consistent estimation is difficult. However, when \(Y\) has finite support and \(\Omega(O;\eta)\) is constructed as a weighted average over possible values \(y\),
\[
\Omega(O;\eta)
=
\sum_{y\in\mathcal{Y}_{A,C}}
\omega_y(O;\eta)\,y,
\]
with weights \(\omega_y(O;\eta)\) constructed as in Appendix~\ref{app:weights}, these properties can be established directly. We focus on these results for the remainder of the paper.

We give an observed-data influence function below:

\begin{theorem}[Observed-Data Influence Function]\label{thm:obsIF}
Under Assumptions~\ref{assump:positivity}--\ref{assump:weight}, the function \(\phi_{\mathrm{obs}}(O; \psi_a, \eta)\) below is an observed-data influence function for the average counterfactual outcome \(\psi_a = \mathbb{E}[Y(a)]\). 

A proof is given in Appendix~\ref{app:obsIF}. 

\begin{align*}
\phi_{\mathrm{obs}}(O; \psi_a, \eta)
&= \frac{\mathbb{I}(A=a)}{p(A \mid C)} 
   \Bigg\{ 
      \Omega(O;\eta)
      - \mathbb{E}[Y \mid A=a, C]
   \Bigg\} + \mathbb{E}[Y \mid A=a, C] - \psi_a.
\end{align*}

When $Y$ has finite support, we take
\[
\Omega(O;\eta)
=
\sum_{y \in \mathcal{Y}_{A,C}}
\omega_y(O;\eta)\, y,
\]
where \(\mathcal{Y}_{A,C}\) denotes the support of \(Y\mid A,C\) and the weights \(\omega_y(O;\eta)\) are constructed as in Appendix~\ref{app:weights}.
\end{theorem}

\subsection{Influence Function Properties}

When $Y$ has finite support, our observed-data influence function gives rise to an estimator that is multiply robust (consistent despite partial model misspecification), formalized below.

\begin{theorem}[Multiple Robustness]\label{thm:robustness}
Let \(\phi_{\mathrm{obs}}(O; \psi_a, \eta)\) denote the observed-data influence function in Theorem~\ref{thm:obsIF}, where $Y$ has finite support and weights \(\omega_y(O;\eta)\) are constructed as in Appendix~\ref{app:weights}. Then the estimator \(\hat{\psi}_a\) constructed from Equation~\ref{eq:estimation} is consistent under standard regularity conditions if any one of nuisance sets (i)--(vi) is correctly specified.

A proof is provided in Appendix~\ref{app:robustness}.

\renewcommand{\theenumi}{\roman{enumi}}%
\begin{enumerate}
  \item $p(A \mid C)$, $\mathbf{M}_W(A,C)$, and $\mathbf{M}_Z(A,C)$;  
  \item $p(A \mid C)$, $\mathbf{M}_W(A,C)$, and $\mathbf{M}_V(A,C)$;  
  \item $p(A \mid C)$, $\mathbf{M}_Z(A,C)$, and $\mathbf{M}_V(A,C)$;  
  \item $\mathbb{E}[Y \mid A,C]$, $\mathbf{M}_W(A,C)$, and $\mathbf{M}_Z(A,C)$;  
  \item $\mathbb{E}[Y \mid A,C]$, $\mathbf{M}_W(A,C)$, and $\mathbf{M}_V(A,C)$;  
  \item $\mathbb{E}[Y \mid A,C]$,  $\mathbf{M}_Z(A,C)$, and $\mathbf{M}_V(A,C)$.  
\end{enumerate}

For \(X \in \{W,Z,V\}\), the matrix \(\mathbf{M}_X(A,C)\) is correctly specified whenever its entries are
correctly specified. In settings where \(X\) has finite support,
this is equivalent to correct specification of the conditional distribution
\(p(X \mid Y,A,C)\).

\end{theorem}

Another key advantage of influence function-based estimators is that they can achieve \(\sqrt{n}\)-consistency and asymptotic normality even when the nuisance functions are not individually estimated at \(\sqrt{n}\) rates. Instead, their asymptotic behavior is governed by higher-order terms involving products of nuisance estimation errors \citep{KennedyTutorial}. When \(Y\) has finite support, we state the corresponding result for our observed-data influence function below.

\begin{theorem}[Asymptotic Normality]\label{thm:efficiency}
Let \(\phi_{\mathrm{obs}}(O; \psi_a, \eta)\) denote the observed-data influence function in Theorem~\ref{thm:obsIF}, where $Y$ has finite support and the weights \(\omega_y(O;\eta)\) are constructed as in Appendix~\ref{app:weights}. If conditions (i)--(v) hold, then, under standard regularity conditions, including sample splitting, the estimator \(\hat{\psi}_a\) obtained from Equation~\ref{eq:estimation} satisfies
\[
\sqrt{n}\,(\hat{\psi}_a - \psi_a)
\;\;\overset{d}{\longrightarrow}\;\;
\mathcal{N}\!\Big(0, \,\mathbb{E}\big[\phi_{\mathrm{obs}}^2(O; \psi_a, \eta)\big]\Big).
\]

We provide a proof in Appendix~\ref{app:efficiency}.

\renewcommand{\theenumi}{\roman{enumi}}%
\begin{enumerate}
  \item $\lVert \hat{p}(A \mid C) - p(A \mid C) \rVert = o_p(1)$, 
  $\lVert \hat{\mathbb{E}}(Y \mid A, C) - \mathbb{E}(Y \mid A, C) \rVert = o_p(1)$, \\ 
  $\left\lVert
\widehat{\mathbf M}_W(A,C)-\mathbf M_W(a,c)
\right\rVert
=
o_p(1)$,  
  $\left\lVert
\widehat{\mathbf M}_Z(A,C)-\mathbf M_Z(A,C)
\right\rVert
=
o_p(1)$, \\
 $\left\lVert
\widehat{\mathbf M}_V(A,C)-\mathbf M_V(A,C)
\right\rVert
=
o_p(1)$;
 \item $\lVert \hat{p}(A \mid C) - p(A \mid C) \rVert$ 
  $\lVert \hat{\mathbb{E}}(Y \mid A, C) - \mathbb{E}(Y \mid A, C) \rVert = o_p(n^{-1/2})$;  
  \item $\left\lVert
\widehat{\mathbf M}_W(A,C)-\mathbf M_W(A,C)
\right\rVert$ 
  $\left\lVert
\widehat{\mathbf M}_Z(A,C)-\mathbf M_Z(A,C)
\right\rVert = o_p(n^{-1/2})$; 
  \item $\left\lVert
\widehat{\mathbf M}_W(A,C)-\mathbf M_W(A,C)
\right\rVert$ 
  $\left\lVert
\widehat{\mathbf M}_V(A,C)-\mathbf M_V(A,C)
\right\rVert = o_p(n^{-1/2})$;  
  \item $\left\lVert
\widehat{\mathbf M}_Z(A,C)-\mathbf M_Z(A,C)
\right\rVert$ 
  $\left\lVert
\widehat{\mathbf M}_V(A,C)-\mathbf M_V(A,C)
\right\rVert = o_p(n^{-1/2})$.

\end{enumerate}

\end{theorem}

The multiple robustness and asymptotic normality results in Theorems~\ref{thm:robustness} and~\ref{thm:efficiency} extend directly to the average treatment effect estimator \(\hat\psi = \hat\psi_1 - \hat\psi_0\), rendering it consistent under proper specification of any of the nuisance sets in Theorem~\ref{thm:robustness}, and asymptotically normal with asymptotic variance
\(
\mathbb{E}\Big[\big\{\phi_{\mathrm{obs}}(O; \psi_1, \eta) - \phi_{\mathrm{obs}}(O; \psi_0, \eta)\big\}^2\Big]
\) under the same conditions as Theorem~\ref{thm:efficiency}.

\section{Simulation}
We evaluate the performance of the influence function-based estimator for the average causal effect under two simulated data-generating processes (DGPs): 
(1) a fully binary setting; (2) a setting with continuous confounders \(C\), continuous proxies (\(W\), \(Z\), and \(V\)), and a binary hidden outcome \(Y\); and (3) a setting with continuous confounders \(C\), continuous proxies (\(W\), \(Z\), and \(V\)), and a three-category hidden outcome \(Y\). We select these settings, where $Y$ has finite support, to demonstrate the finite-sample behavior of our proposed influence function-based estimators when \(\Omega(O;\eta)\) admits a relatively tractable characterization.

\subsection{Binary Data-Generating Process}\label{subsec:binary_dgp}
The fully binary data-generating process we simulate is described in Appendix~\ref{app:binary_dgp}.  
We compare three estimators:
\begin{enumerate}
   \item \textbf{Proposed estimator:} computes \(\hat{\psi} = \hat{\psi}_1 - \hat{\psi}_0\), where each \(\hat{\psi}_a\) is obtained by solving Equation~\ref{eq:estimation} using the observed-data influence function \(\phi_{\mathrm{obs}}(O; \psi_a, \eta)\) from Theorem~\ref{thm:obsIF}, with weights \(\omega_y(O; \eta)\) constructed in the same manner as in Example~\ref{ex:binary}. Nuisance functions appearing in \(\phi_{\mathrm{obs}}(O; \psi_a, \eta)\) and \(\omega_y(O; \eta)\), including \(\mathbb{E}[Y \mid A,C]\) and proxy regression functions \(\mathbb{E}[W \mid Y]\), \(\mathbb{E}[Z \mid Y]\), and \(\mathbb{E}[V \mid Y]\), are estimated using procedures described in Appendix~\ref{app:binary_dgp}. 
    \item \textbf{Oracle estimator:} has direct access to \(Y\) and computes \(\hat\psi = \hat\psi_1 - \hat\psi_0\), where each \(\hat\psi_a\) is obtained by solving
    \[
    \mathbb{P}_n\{\phi_{\mathrm{full}}(A,Y,C; \hat\psi_a, \hat\eta)\} = 0,
    \]
    with \(\phi_{\mathrm{full}}(A,Y,C; \psi_a, \eta)\) given in Equation~\ref{eq:full_IF};

    \item \textbf{Naïve estimator:} mirrors the oracle approach but substitutes the unobserved outcome \(Y\) with the proxy \(W\), and substitutes the nuisance function \(\mathbb{E}[Y \mid A=a, C]\) with \(\mathbb{E}[W \mid A=a, C]\) in \(\phi_{\mathrm{full}}(A,Y,C; \psi_a, \eta)\) given in Equation~\ref{eq:full_IF}.
\end{enumerate}
For each, we report the mean estimate, bias, and scaled variance (\(n \times \mathrm{Var}\)).

Results appear in Table~\ref{tab:sim_results_binary}. The naïve method substantially underestimates the target parameter across sample sizes. The proposed estimator eliminates bias. This correction comes at the cost of higher variance compared to the oracle (scaled variance around 1.7 vs 0.6), reflecting the efficiency loss incurred when the latent outcome must be ``reconstructed'' from observed proxies through the proposed weighting procedure.

\begin{table}[ht]
\centering
\large
\caption{ \large Binary DGP}
\label{tab:sim_results_binary}
\begin{tabular}{lcccc}
\hline
Method & Sample Size ($n$) & Mean & Bias & $n \times$ Variance \\
\hline
Proposed
        & 200  & 0.591 &  -0.001 & 1.732 \\
        & 500  & 0.592 & -0.001 & 1.671 \\
         & 1000 & 0.591 & -0.002 & 1.734 \\
\hline
Naïve & 200  & 0.475 & -0.118 & 0.796 \\
        & 500  & 0.472 & -0.120 & 0.748 \\
         & 1000 & 0.474 & -0.119 & 0.708 \\
\hline
Oracle & 200  & 0.593 & 0.001 & 0.599 \\
        & 500  & 0.593 & 0.001 & 0.582 \\
         & 1000 & 0.592 & 0.000 & 0.609 \\

\hline
\end{tabular}
\end{table}

\subsection{Mixed DGP 1 (Binary Outcome; Continuous Confounders and Proxies)} 

 The first mixed data-generating process (with a binary outcome, and continuous confounder and proxies) we simulate is described in Appendix~\ref{app:mixed_dgp1}. We construct estimators similar to those considered in Section~\ref{subsec:binary_dgp}. For our proposed influence function-based estimator, the weights
\(\omega(O;\eta)\) are constructed in the same manner as Example~\ref{ex:binary}, and nuisance functions are
estimated using the procedure described in Appendix~\ref{app:mixed_dgp1}.

Results appear in Table~\ref{tab:sim_results_cont1}, with similar trends to the binary data-generating process, but higher variance for the proposed estimator. The proposed estimator already exhibits little bias (0.015) at sample size \(n=200\), although \(n \times \mathrm{Var}\) is large (approximately \(132.6\) at \(n=200\)), decreasing considerably at \(n=500\) and \(n=1000\) (to \(22.9\) and \(25.5\) , respectively).

\begin{table}[ht]
\centering
\large
\caption{ \large Binary $Y$, Continuous \(C,W,Z,V\)}
\label{tab:sim_results_cont1}
\begin{tabular}{lcccc}
\hline
Method & Sample Size ($n$) & Mean & Bias & $n \times$ Variance \\
\hline
Proposed
        & 200  & 0.463 &  0.015 & 132.602 \\
        & 500  & 0.445 & -0.002 & 22.897 \\
         & 1000 & 0.441 & -0.007 & 25.533 \\
\hline
Naïve & 200  & 0.681 & 0.234 & 7.211 \\
        & 500  & 0.676 & 0.229 & 6.796 \\
         & 1000 & 0.669 & 0.222 & 7.121 \\
\hline
Oracle & 200  & 0.447 & 0.000 & 0.773 \\
        & 500  & 0.447 & 0.000 & 0.821 \\
         & 1000 & 0.447 & -0.001 & 0.803 \\

\hline
\end{tabular}
\end{table}

\subsection{Mixed DGP 2 (Three-category Outcome; Continuous Confounders and Proxies)} The second mixed data-generating process (with a three-category outcome, and continuous confounder and proxies) that we simulate is described in Appendix~\ref{app:mixed_dgp2}. We construct estimators similar to those considered in Section~\ref{subsec:binary_dgp}. For our proposed influence function-based estimator, the weights
\(\omega(O;\eta)\) are constructed analogously to Example~\ref{ex:binary}
using three polynomial basis functions, and nuisance functions are
estimated using the procedure described in Appendix~\ref{app:mixed_dgp2}.

Results appear in Table~\ref{tab:sim_results_cont2}. At \(n=200\), the proposed estimator has bias \(0.063\) and \(n\times\mathrm{Var}\) around \(916.2\). Performance improves markedly with increasing sample size. At \(n=500\), bias decreases to \(-0.001\) and \(n \times \mathrm{Var}\) decreases to around \(11.8\). Similar results are observed at \(n=1000\), where bias is \(0.002\) and \(n \times \mathrm{Var}\) approximately \(11.9\).

\begin{table}[ht]
\centering
\large
\caption{ \large Three-category $Y$, Continuous $C,W,Z,V$}
\label{tab:sim_results_cont2}
\begin{tabular}{lcccc}
\hline
Method & Sample Size ($n$) & Mean & Bias & $n \times$ Variance \\
\hline
Proposed
        & 200  & 0.602 &  0.063 & 916.194 \\
        & 500  & 0.538 & -0.001 & 11.804 \\
         & 1000 & 0.541 & 0.002 & 11.895 \\
\hline
Naïve & 200  & 2.316 & 1.776 & 52.231 \\
        & 500  & 2.337 & 1.797 & 45.914 \\
         & 1000 & 2.317 & 1.777 & 48.849 \\
\hline
Oracle & 200  & 0.536 & -0.004 & 2.399 \\
        & 500  & 0.541 & 0.001 & 2.252 \\
         & 1000 & 0.538 & -0.002 & 2.362 \\

\hline
\end{tabular}
\end{table}
\section{Conclusion}

In this work, we study causal inference in the presence of hidden variables through a proxy-based lens. Leveraging mild structural assumptions, we show that it is possible to recover the full-data law from the observed-data law. Building on this identification strategy, we develop an influence function-based approach for estimating causal effects, yielding examples of estimators that are $\sqrt{n}$-consistent and exhibit multiple robustness. Notably, our framework avoids reliance on unbiased proxies, partial observation, or strong parametric restrictions.

A number of key directions remain to be explored. First, the multiple robustness and product-rate bias properties established here under outcomes $Y$ with finite support have not yet been developed for the continuous outcome setting. Second, while our framework assumes three proxies, many applications feature multiple candidate proxies. Understanding how to optimally combine information across these proxies for efficient estimation remains an open problem. Third, developing principled modeling strategies and learning procedures that reliably attain product rates is an important challenge, particularly in high-dimensional or nonparametric settings. Finally, our results hint at a more general theory for influence functions under proxy-based full-data law identification, including systematic construction of influence function-based estimators across a broad class of target functionals. Addressing these questions would strengthen the connection between identification theory and practical, scalable estimation.

\bibliographystyle{imsart-nameyear}
\bibliography{references}

\newpage
\appendix

\section{Proof of Full-data Law Identification~\ref{thm:obsIF}}\label{app:identification}

Under Assumptions~\ref{assump:boundedness_hidden_Y}-\ref{assump:unique_mapping}, \(p(Y,W,Z,V \mid A,C)\) is identified by Theorem 1 in \citet{hu2008instrumental}. Then, since \(p(A,C)\) is identified from the observed data distribution, the full-data law of hidden and observed variables \(p(A,Y,C,W,Z,V)\) is identified.

\section{Proof of Proposition~\ref{prop:guaranteed_weight}}\label{app:guaranteed_weight}

\begin{proof}
By Assumptions~\ref{assump:completeness_hidden_Y} and the strengthening of Assumption~\ref{assump:distinctness_hidden_Y}, for each \((a,c)\), the conditional laws in
\(\{
p(W \mid Y=y, A=a, C = c):
y\}
\)
are linearly independent. 

Let \(R(A,C)\), denoted in shorthand by \(R\), be the number of support points in
\(\mathcal Y_{A,C}\), and define the dominating measure
\[
\nu_{Y,a,c}
=
\sum_{i=1}^{R(a,c)}
p(W \mid Y=y_i,A=a,C=c).
\]
Let
\[
q_{W,i}(W,a,c)
=
\frac{
dp(W \mid Y=y_i,a,c)
}{
d\nu_{Y,a,c}
}(W)
\]
denote the \(i\)th Radon--Nikodym derivative. Write
\[
q_W(W,a,c)
=
\begin{pmatrix}
q_{W,1}(W,a,c)\\
\vdots\\
q_{W,R}(W,a,c)
\end{pmatrix},
\]
and define
\[
\mathbf{M}_W(a,c)
=
\int
q_W(w,a,c)q_W(w,a,c)^\top
\,d\nu_{Y,a,c}(w).
\]

Because the conditional laws in
\(\{
p(W \mid Y=y, A=a, C = c):
y\}
\)
are linearly independent, the Radon--Nikodym derivatives in \(\{
q_{W,i}(W,a,c):i\}
\) are also linearly independent. Hence, \(\mathbf{M}_W(\vec a,c)\) is invertible.

Now define
\[
h_{W,i}(W,a,c; \eta)
=
e_i^\top
\mathbf{M}_W(a,c)^{-1}
q_W(W,a,c),
\]
where \(e_i\) is the \(i\)th standard basis vector. Then

\begin{align*}
&\mathbb E\!\left[
h_{W,i}(W,A,C; \eta)
\,\middle|\,
Y=y_j,a,c
\right] \\
&\qquad=
e_i^\top
\mathbf{M}_W(a,c)^{-1}
\int
q_W(w,a,c)q_{W,j}(w,a,c)
\,d\nu_{Y,a,c}(w)\\
&\qquad=
e_i^\top
\mathbf{M}_W(a,c)^{-1}
\mathbf{M}_W(a,c)
e_j\\
&\qquad=
e_i^\top e_j\\
&\qquad=
\mathbf 1\{i=j\}.
\end{align*}

Similarly, define \(q_Z(Z,A,C)\) to construct \(h_{Z,i}(Z,A,C; \eta)\) and \(q_V(V,A,C)\) to construct \(h_{V,i}(V,A,C; \eta)\). Then combining terms,
\[
\begin{aligned}
&\omega_{y_i}(O;\eta)
=
h_{W,i}(W,A,C;\eta)h_{Z,i}(Z,A,C;\eta)
+
h_{W,i}(W,A,C;\eta)h_{V,i}(V,A,C;\eta) \\
&+
h_{Z,i}(Z,A,C;\eta)h_{V,i}(V,A,C;\eta)
-
2h_{W,i}(W,A,C;\eta)h_{Z,i}(Z,A,C;\eta)h_{V,i}(V,A,C;\eta),
\end{aligned}
\]
satisfies Equation~\ref{eq:indicator} in Proposition~\ref{prop:guaranteed_weight}. Hence,
\(
\Omega(O;\eta)
=
\sum_{y_i\in\mathcal Y_{A,C}}
y_i\,\omega_{y_i}(O;\eta),
\) satisfies Assumption~\ref{assump:weight}.
\end{proof}

\section{General Weight Construction for $Y$ with Finite Support}\label{app:weights}
Let \(R(A,C)\), denoted in shorthand by \(R\), be the number of support points in
\(\mathcal Y_{A,C}\).
Let
\[
q_W(W,A,C)
=
\bigl(
q_{W,1}(W),
\dots,
q_{W,N}(W)
\bigr)^\top
\]
be a vector of \(N \geq R(A,C)\) basis functions of \(W\) such that the conditional moment matrix
\[
\mathbf{M}_W(A,C)
=
\begin{pmatrix}
\mathbb{E}[q_{W,1}(W)\mid y_1,A,C] & \cdots &
\mathbb{E}[q_{W,1}(W)\mid y_R,A,C]
\\
\vdots & & \vdots
\\
\mathbb{E}[q_{W,N}(W)\mid y_1,A,C] & \cdots &
\mathbb{E}[q_{W,N}(W)\mid y_R,A,C]
\end{pmatrix}
\]
is left invertible.

Assumption~\ref{assump:completeness_hidden_Y} implies that the conditional laws in
\[
\{p(W\mid y,A,C): y\in\mathcal{Y}_{A,C}\}
\]
are linearly independent. Hence, when \(W\) has finite support with support
\(\{w_1,\dots,w_N\}\), a natural choice of $q_W(W)$ is
\[
q_W(W)
=
\bigl(
\mathbf{1}\{W=w_1\},
\dots,
\mathbf{1}\{W=w_N\}
\bigr)^\top,
\]
in which case \(\mathbf{M}_W(A,C)\) coincides with the matrix whose columns are linearly independent conditional distributions, so it is left invertible. When \(W\) is continuous, \(q_W(W,A,C)\) may be constructed from polynomial, spline, or other function classes.

Define
\[
h_W(W,A,C; \eta)
=
\mathbf{M}_W(A,C)^{-1}q_W(W,A,C),
\]
where
\[ \mathbf{M}_W(A,C)^{-1} \text{ is the left inverse of } \mathbf{M}_W(A,C).
\]

Since \(\mathbb{E}[q_W(W,A,C)\mid Y=y_j, A,C]
=
\mathbf M_W(A,C)e_j,
\) where \(e_j\) denotes the \(j\)-th standard basis vector,
it follows that the \(i\)-th component of \(h_W(W,A,C; \eta)\), denoted
\(h_{W,A,C,i}(W; \eta)\), satisfies
\[
\mathbb E[h_{W,i}(W,A,C; \eta)\mid Y=y_j,A,C]
=
\mathbf 1\{i=j\}.
\]

Similarly, define \(q_Z(Z,A,C)\) to construct \(h_{Z,i}(Z,A,C; \eta)\) and \(q_V(V,A,C)\) to construct \(h_{V,i}(V,A,C; \eta)\). Then combining terms,
\[
\begin{aligned}
&\omega_{y_i}(O;\eta)
=
h_{W,i}(W,A,C;\eta)h_{Z,i}(Z,A,C;\eta)
+
h_{W,i}(W,A,C;\eta)h_{V,i}(V,A,C;\eta) \\
&+
h_{Z,i}(Z,A,C;\eta)h_{V,i}(V,A,C;\eta)
-
2h_{W,i}(W,A,C;\eta)h_{Z,i}(Z,A,C;\eta)h_{V,i}(V,A,C;\eta),
\end{aligned}
\]
satisfies Equation~\ref{eq:indicator} in Proposition~\ref{prop:guaranteed_weight}. Hence,
\(
\Omega(O;\eta)
=
\sum_{y_i\in\mathcal Y_{A,C}}
y_i\,\omega_{y_i}(O;\eta),
\) satisfies Assumption~\ref{assump:weight}.

\section{Proof of Proposition~\ref{prop:polynomial_basis}}\label{app:polynomial_basis}

In the location-shift model of Proposition~\ref{prop:polynomial_basis}, the polynomial basis provides an explicit choice of \(q_X(X,A,C)\) for each \(X \in \{W,Z,V\}\). Indeed, since
\[
X \mid Y=y_i,A,C
=
\mu_{X,i}(A,C)+\epsilon_X(A,C),
\]
where \(\epsilon_X(A,C)\) does not depend on \(Y\), then 
\[
\mathbb E[X^k \mid Y=y_i,A,C] = \mathbb{E}[\{\mu^X_i(A,C) + \epsilon_X(A,C)\}^k].
\] 

The conditional moment matrix \(\mathbf M_X(A,C)\) in Appendix~\ref{app:weights}, formed using the basis
\[
q_X(X,A,C)
=
(1,X,\ldots,X^{R-1})^\top,
\]
has entries
\[
\mathbf M_X(A,C)_{\ell i}
=
\mathbb E[X^{\ell-1}\mid Y=y_i,A,C], \quad \ell = 1, \cdots, R.
\]

By the binomial expansion, \(\mathbf M_X(A,C)\) is an invertible transformation of the Vandermonde matrix 
\[
\mathbf V_X(A,C)
=
\begin{pmatrix}
1 & 1 & \cdots & 1\\
\mu_1^X(A,C) & \mu_2^X(A,C) & \cdots & \mu_R^X(A,C)\\
\{\mu_1^X(A,C)\}^2 & \{\mu_2^X(A,C)\}^2 & \cdots & \{\mu_R^X(A,C)\}^2\\
\vdots & \vdots & & \vdots\\
\{\mu_1^X(A,C)\}^{R-1} &
\{\mu_2^X(A,C)\}^{R-1} &
\cdots &
\{\mu_R^X(A,C)\}^{R-1}
\end{pmatrix}.
\]

Since the values
\[
\mu_1^X(A,C),\ldots,\mu_R^X(A,C)
\]
are distinct, \(\mathbf V_X(A,C)\) has full rank \(R\). Since invertible linear transformations preserve rank, \(\mathbf M_X(A,C)\) also has full rank \(R\), and hence is left invertible.

 Therefore the functions \(h_{W,i}(W,A,C;\eta)\), \(h_{Z,i}(Z,A,C;\eta)\), and \(h_{V,i}(V,A,C;\eta)\) in Appendix~\ref{app:weights} can be constructed setting \(
q_X(X,A,C)
=
(1,X,\ldots,X^{R-1})^\top
\) for each \(X \in \{W,Z,V\}\). These functions constitute \(\omega_{y_i}(O; \eta)\) satisfying Equation~\ref{eq:indicator} in Proposition~\ref{prop:guaranteed_weight}. 

\section{Proof of Full-data Influence Function}\label{app:full_IF}
\begin{proof}
Let $\{P_\epsilon(H) : \epsilon \in \mathbb{R}\}$ be a regular parametric submodel of the full-data law with corresponding densities \(p_\epsilon(H=(A,Y,C,W,Z,V))\) and score
\[
s(H) = \left.\frac{\partial}{\partial \epsilon} \log p_\epsilon(H)\right|_{\epsilon=0}.
\]

Under Assumptions~\ref{assump:positivity}-\ref{assump:ignorability}, we aim to show $\phi_\mathrm{full}(A,Y,C; \psi_a, \eta)$ from Equation~\ref{eq:full_IF} is a full-data influence function for target \(\psi_a = \mathbb{E}[Y(a)] = \int p(Y|A=a,C=c)p(C=c)dc\), i.e., that 
\begin{equation}\label{eq:full_if_proof}
\left.\frac{d}{d\epsilon}\psi_a(P_\epsilon(H))\right|_{\epsilon=0}
= \mathbb{E}\big[\phi_{\mathrm{full}}(A,Y,C; \psi_a, \eta)\, s(H)\big],
\qquad 
\mathbb{E}\big[\phi_{\mathrm{full}}(A,Y,C; \psi_a, \eta)\big]=0.
\end{equation}

The regular parametric submodel $\{P_\epsilon(H) : \epsilon \in \mathbb{R}\}$ of the full-data law induces a corresponding regular parametric submodel $\{P_\epsilon(A,Y,C) : \epsilon \in \mathbb{R}\}$ of the marginal law of \((A,Y,C)\) via marginalization over \((W,Z)\). The induced submodel $\{P_\epsilon(A,Y,C) : \epsilon \in \mathbb{R}\}$ is indexed by the same parameter \(\epsilon\), with score
\[
s(A,Y,C)
=
\left.
\frac{\partial}{\partial \epsilon}
\log p_\epsilon(A,Y,C)
\right|_{\epsilon=0}
=
\mathbb{E}[s(H)\mid A,Y,C].
\]

Then because \(\psi_a\) is a functional of the marginal law, 

\[
\left.\frac{d}{d\epsilon}\psi_a(P_\epsilon(H))\right|_{\epsilon=0}
=\left.\frac{d}{d\epsilon}\psi_a(P_\epsilon(A,Y,C))\right|_{\epsilon=0}.
\]

Since \(\phi_\mathrm{full}(A,Y,C; \psi_a, \eta)\) is an influence function in the marginal model \citep{KennedyTutorial, robins94estimation}, 
\[
\begin{aligned}
\left.\frac{d}{d\epsilon}\psi_a(P_\epsilon(H))\right|_{\epsilon=0}
&=\left.\frac{d}{d\epsilon}\psi_a(P_\epsilon(A,Y,C))\right|_{\epsilon=0}
= \mathbb{E}\big[
\phi_{\mathrm{full}}(A,Y,C; \psi_a,\eta)\,
s(A,Y,C)
\big] \\
&= \mathbb{E}\big[\phi_{\mathrm{full}}(A,Y,C)\,\mathbb{E}\{s(Y\mid A,C) \mid Y,A,C\}\big]  \\
&\quad + \mathbb{E}\big[\phi_{\mathrm{full}}(A,Y,C)\,\mathbb{E}\{s(A\mid C) \mid Y,A,C\}\big] \\
&\quad + \mathbb{E}\big[\phi_{\mathrm{full}} (A,Y,C)\,\mathbb{E}\{s(C)\mid Y,A,C\}\big] \\
&\quad + \mathbb{E}\Big[
\phi_{\mathrm{full}}(A,Y,C)\,
\mathbb{E}\{s(W\mid Y,A,C)\mid Y,A,C\}
\Big] \\
&\quad + \mathbb{E}\Big[
\phi_{\mathrm{full}}(A,Y,C)\,
\mathbb{E}\{s(Z\mid Y,A,C)\mid Y,A,C\}
\Big] \\
&\quad + \mathbb{E}\Big[
\phi_{\mathrm{full}}(A,Y,C)\,
\mathbb{E}\{s(V\mid Y,A,C)\mid Y,A,C\}
\Big] \\
&= \mathbb{E}\big[\phi_{\mathrm{full}}(A,Y,C)\,s(Y\mid A,C)\big]  + \mathbb{E}\big[\phi_{\mathrm{full}}(A,Y,C)\,s(A\mid C)\big] \\
&\quad + \mathbb{E}\big[\phi_{\mathrm{full}}(A,Y,C)\,s(C)\big] \\
&\quad + \mathbb{E}\big[\phi_{\mathrm{full}}(A,Y,C)\,s(W\mid Y,A,C)\big] \\
&\quad + \mathbb{E}\big[\phi_{\mathrm{full}}(A,Y,C)\,s(Z\mid Y,A,C)\big] \\
&\quad + \mathbb{E}\big[\phi_{\mathrm{full}}(A,Y,C)\,s(V\mid Y,A,C)\big] \\
&= \mathbb{E}\big[
\phi_{\mathrm{full}}(A,Y,C; \psi_a,\eta)\,
s(H)
\big].
\end{aligned}
\]
satisfying the first part of Equation~\ref{eq:full_if_proof}.

We know the second part of Equation~\ref{eq:full_if_proof} from former results \citep{KennedyTutorial, robins94estimation}.
\end{proof}
\section{Proof of Theorem~\ref{thm:obsIF}}\label{app:obsIF}
\begin{proof}

Let $\{P_\epsilon(O) : \epsilon \in \mathbb{R}\}$ be a regular parametric submodel of the observed-data law with score
\[
s(O) = \left.\frac{\partial}{\partial \epsilon} \log p_\epsilon(O)\right|_{\epsilon=0}.
\]

Under Assumptions~\ref{assump:positivity}-\ref{assump:weight}, we aim to show $\phi_\mathrm{obs}(O; \psi_a, \eta)$ from Theorem~\ref{thm:obsIF} is an observed-data influence function for target \(\psi_a = \mathbb{E}[Y(a)]\), i.e., that
\begin{equation}\label{eq_obs_if_proof}
\left.\frac{d}{d\epsilon}\psi_a(P_\epsilon(O))\right|_{\epsilon=0}
= \mathbb{E}\big[\phi_{\mathrm{obs}}(O; \psi_a, \eta)\, s(O)\big],
\qquad 
\mathbb{E}\big[\phi_{\mathrm{obs}}(O; \psi_a, \eta)\big]=0.
\end{equation}

The regular parametric submodel
\(\{P_\epsilon(O) : \epsilon \in \mathbb{R}\}\)
of the observed-data law is induced by a regular parametric submodel
\(\{P_\epsilon(H) : \epsilon \in \mathbb{R}\}\)
of the full-data law through marginalization of \(Y\). The full-data submodel is indexed by the same parameter \(\epsilon\), with score
\[
s(H)
=
\left.
\frac{\partial}{\partial \epsilon}
\log p_\epsilon(H)
\right|_{\epsilon=0},
\]
such that
\[
s(O)
=
\mathbb{E}[s(H)\mid O].
\]

Then because \(\psi_a\) is a functional of the observed-data law, 
\[\left.\frac{d}{d\epsilon}\psi_a(P_\epsilon(O))\right|_{\epsilon=0}
= \left.\frac{d}{d\epsilon}\psi_a(P_\epsilon(H))\right|_{\epsilon=0}.
\]
Since \(\phi_{\mathrm{full}}(H; \psi_a, \eta)\) from Equation~\ref{eq:full_IF} is an influence function in the full-data model,
\begin{align*}
\left.\frac{d}{d\epsilon}\psi_a(P_\epsilon(O))\right|_{\epsilon=0}&= \left.\frac{d}{d\epsilon}\psi_a(P_\epsilon(H))\right|_{\epsilon=0}= \mathbb{E}\big[\phi_{\mathrm{full}}(H; \psi_a, \eta)\, s(H)\big] \\ 
&= \mathbb{E}\big[\phi_{\mathrm{full}}(H)\, s(Y \mid A, C)\big] 
+ \mathbb{E}\big[\phi_{\mathrm{full}}(H)\, s(A \mid C)\big] + \mathbb{E}\big[\phi_{\mathrm{full}}(H)\, s(C)\big] \\
&\quad + \mathbb{E}\big[\phi_{\mathrm{full}}(H)\, s(W \mid Y, A, C)\big] 
+ \mathbb{E}\big[\phi_{\mathrm{full}}(H)\, s(Z \mid Y, A, C)\big] \\
&\quad
+ \mathbb{E}\big[\phi_{\mathrm{full}}(H)\, s(V \mid Y, A, C)\big] \\
&= \mathbb{E}\!\left[ \mathbb{E}\big[\phi_{\mathrm{full}}(H) \mid Y, A, C \big] s(Y \mid A, C) \right] 
+ \mathbb{E}\!\left[ \mathbb{E}\big[\phi_{\mathrm{full}}(H) \mid A, C \big] s(A \mid C) \right] \\
&\quad + \mathbb{E}\!\left[ \mathbb{E}\big[\phi_{\mathrm{full}}(H) \mid C \big] s(C) \right] 
+ \mathbb{E}\!\left[ \mathbb{E}\big[\phi_{\mathrm{full}}(H) \mid W, Y, A, C \big] s(W \mid Y, A, C) \right] \\
&\quad + \mathbb{E}\!\left[ \mathbb{E}\big[\phi_{\mathrm{full}}(H) \mid Z, Y, A, C \big] s(Z \mid Y, A, C) \right] \\
&\quad + \mathbb{E}\!\left[ \mathbb{E}\big[\phi_{\mathrm{full}}(H) \mid V, Y, A, C \big] s(V \mid Y, A, C) \right].
\end{align*}

Therefore, since the following conditional expectations hold:
\begin{align}
\mathbb{E}[\phi_{\mathrm{obs}}(O) \mid Y, A, C] &= \mathbb{E}[\phi_{\mathrm{full}}(H) \mid Y, A, C], \\
\mathbb{E}[\phi_{\mathrm{obs}}(O) \mid A, C] &= \mathbb{E}[\phi_{\mathrm{full}}(H) \mid A, C], \\
\mathbb{E}[\phi_{\mathrm{obs}}(O) \mid C] &= \mathbb{E}[\phi_{\mathrm{full}}(H) \mid C], \\
\mathbb{E}[\phi_{\mathrm{obs}}(O) \mid W, Y, A, C] &= \mathbb{E}[\phi_{\mathrm{full}}(H) \mid W, Y, A, C], \\
\mathbb{E}[\phi_{\mathrm{obs}}(O) \mid Z, Y, A, C] &= \mathbb{E}[\phi_{\mathrm{full}}(H) \mid Z, Y, A, C], \\
\mathbb{E}[\phi_{\mathrm{obs}}(O) \mid V, Y, A, C] &= \mathbb{E}[\phi_{\mathrm{full}}(H) \mid V, Y, A, C],
\end{align}
we can substitute these into the earlier expression to obtain
\begin{align*}
\left.\frac{d}{d\epsilon}\psi_a(P_\epsilon)\right|_{\epsilon=0} 
&= \mathbb{E}\big[ \phi_{\mathrm{obs}}(O)\, s(Y \mid A, C) \big] 
+ \mathbb{E}\big[ \phi_{\mathrm{obs}}(O)\, s(A \mid C) \big] \\
&\quad + \mathbb{E}\big[ \phi_{\mathrm{obs}}(O)\, s(C) \big]
+ \mathbb{E}\big[ \phi_{\mathrm{obs}}(O)\, s(W \mid Y, A, C) \big] \\
&\quad + \mathbb{E}\big[ \phi_{\mathrm{obs}}(O)\, s(Z \mid Y, A, C) \big]
+ \mathbb{E}\big[ \phi_{\mathrm{obs}}(O)\, s(V \mid Y, A, C) \big] \\
&= \mathbb{E}\big[ \phi_{\mathrm{obs}}(O)\, s(H) \big].
\end{align*}

Since \(\phi_{\mathrm{obs}}(O)\) is a function of the observed data,
\begin{align*}
\left.\frac{d}{d\epsilon}\psi_a(P_\epsilon)\right|_{\epsilon=0}
&= \mathbb{E}\big[\phi_{\mathrm{obs}}(O)\, s(H)\big] \\
&= \mathbb{E}\big[ \phi_{\mathrm{obs}}(O)\, \mathbb{E}[s(H) \mid O] \big] \\
&= \mathbb{E}\big[ \phi_{\mathrm{obs}}(O)\, s(O) \big],
\end{align*}
satisfying the first part of Equation~\ref{eq_obs_if_proof}.

Next, note \(
\mathbb{E}[\Omega(O;\eta) \mid Y,A,C] = Y.\)
Therefore,
\begin{align*}
\mathbb{E}[\phi_\mathrm{obs}(O; \psi_a, \eta)] 
&= \mathbb{E}\Bigg[
   \mathbb{E}\Bigg[
      \frac{\mathbb{I}(A=a)}{p(A \mid C)} 
      \Big\{
        \Omega(O; \eta)
         - \mathbb{E}[Y \mid A=a,C]\Big\} \\
&\hspace{1cm}
      + \mathbb{E}[Y \mid A=a,C] - \psi_{a}
      \;\Bigg|\; Y=y,A=a,C=c
   \Bigg]
\Bigg] \\
&= \mathbb{E}\left[
   \frac{\mathbb{I}(A=a)}{p(A \mid C)} 
   \{\, Y - \mathbb{E}[Y \mid A=a,C] \,\}
\right]+ \mathbb{E}[Y \mid A=a,C] - \psi_{a} \\
&= \mathbb{E}[\phi_\mathrm{full}(A,Y,C; \psi_a, \eta)] = 0,
\end{align*}
satisfying the second part of Equation~\ref{eq_obs_if_proof}.
\end{proof}

\section{Proof of Theorem~\ref{thm:robustness}}\label{app:robustness}
\begin{proof}
Consider case (i), where \(p(A \mid C)\), $\mathbf{M}_W(A,C)$, and $\mathbf{M}_Z(A,C)$ are correctly specified.  

Since $\mathbf{M}_W(A,C)$ and $\mathbf{M}_Z(A,C)$ are correctly specified, and Assumption~\ref{assump:ci_hidden_y} holds,
\[
\mathbb{E}[\Omega(O;\eta) \mid Y,A,C] =\mathbb{E}[\sum_{y \in \mathcal{Y}_{A,C}} \omega_y(O; \eta) \Big(
         y \Big)\mid Y,A,C] = Y.\] 

Therefore,
\begin{align*}
\mathbb{E}[\phi_\mathrm{obs}(O; \psi_a, \eta)] 
&= \mathbb{E}\Bigg[
   \mathbb{E}\Bigg[
      \frac{\mathbb{I}(A=a)}{p(A \mid C)} 
      \Big\{
         \Omega(O;\eta)
         - \mathbb{E}[Y \mid A=a,C]\Big\} \\
&\hspace{1cm}
      + \mathbb{E}[Y \mid A=a,C] - \psi_{a}
      \;\Bigg|\; Y=y,A=a,C=c
   \Bigg]
\Bigg] \\
&= \mathbb{E}\left[
   \frac{\mathbb{I}(A=a)}{p(A \mid C)} 
   \{\, Y - \mathbb{E}[Y \mid A=a,C] \,\}
\right]+ \mathbb{E}[Y \mid A=a,C] - \psi_{a} \\
&= \mathbb{E}[\phi_\mathrm{full}(A,Y,C; \psi_a, \eta)] = 0,
\end{align*}
where the last line follows by \(p(A \mid C)\) being correctly specified and the double robustness property of \(\phi_\mathrm{full}(A,Y,C; \psi_a, \eta)\).
The proofs for cases (ii)--(vi) proceed analogously.
\end{proof}

\section{Proof of Theorem~\ref{thm:efficiency}}\label{app:efficiency}

\begin{proof}
Let $ \phi = \phi_\mathrm{obs}(O; \psi_a, \eta)$ from Theorem~\ref{thm:obsIF}. Expand \[
\hat{\psi_a} - \psi_a 
= \mathbb{P}_n \hat{\phi} - \mathbb{P}\phi
= (\mathbb{P}_n - \mathbb{P})(\hat{\phi} - \phi) + \mathbb{P}(\hat{\phi} - \phi) + (\mathbb{P}_n - \mathbb{P})\phi
= R_1 + R_2 + (\mathbb{P}_n - \mathbb{P})\phi,
\]
where 
\(
R_1 = (\mathbb{P}_n - \mathbb{P})(\hat{\phi} - \phi), 
\quad 
R_2 = \mathbb{P}(\hat{\phi} - \phi).
\)

$R_1 = o_p(n^{-1/2})$ under consistency of nuisance functions [condition (i)] and sample splitting by Proposition 2 in \citet{kennedy2020sharp}.

Below we show $R_2 = o_p(n^{-1/2})$ under conditions (i)-(vii), therefore
\[
\hat{\psi_a} - \psi_a
= (\mathbb{P}_n - \mathbb{P})\phi + o_p(n^{-1/2})
= \mathbb{P}_n\phi + o_p(n^{-1/2})
\]
which is equivalent to \(
\sqrt{n}\,(\hat{\psi_a}-\psi_a) \;\;\overset{d}{\longrightarrow}\;\; 
\mathcal{N}\!\big(0, \,\mathbb{E}[\phi^2]\big).
\)

First, define
\[
\mu^Y_{a} := \mathbb{E}[Y \mid A=a, C],
\qquad
\pi_{a}  := p(A=a \mid C).
\]

So that      
\begin{align*}
R_2 
&= \mathbb{P}(\hat{\phi} - \phi) = \mathbb{E}(\hat{\phi} - \phi) \\[6pt]
&= \mathbb{E}\!\left(
   \frac{\mathbb{I}(A=a)}{\hat{\pi}_{a}}
   \Big\{
      \sum_{y \in \mathcal{Y}_{A,C}} \omega_y(O; \hat\eta)
      \Big(
         y \Big)
      - \hat{\mu}^Y_{a}
   \Bigg\} 
   + \hat{\mu}^Y_{a} - \psi_a
\right) \\[10pt]
&= \mathbb{E}\!\left(
   \frac{\mathbb{I}(A=a)}{\hat{\pi}_{a}}
   \Bigg\{
      \sum_{y \in \mathcal{Y}_{A,C}} \omega_y(O; \hat\eta)
      \Big(
         y \Big)
      - \mu^Y_{a} + \mu^Y_{a} - \hat{\mu}^Y_{a}
   \Bigg\} 
   + \hat{\mu}^Y_{a} - \psi_a
\right) \\[10pt]
&= \mathbb{E}\!\Bigg(
   \frac{\mathbb{I}(A=a)}{\hat{\pi}_{a}}
   \Big\{
      \sum_{y \in \mathcal{Y}_{A,C}} \omega_y(O; \hat\eta)
      \Big(
         y \Big)
      - \mu^Y_{a}
   \Big\}
   + \frac{\mathbb{I}(A=a)}{\hat{\pi}_{a}}
     \big\{\mu^Y_{a} - \hat{\mu}^Y_{a}\big\}
   + \hat{\mu}^Y_{a} - \psi_a
\Bigg) \\[10pt]
&= \mathbb{E}\!\Bigg(
   \frac{\mathbb{I}(A=a)}{\hat{\pi}_{a}}
   \Big\{
      \sum_{y \in \mathcal{Y}_{A,C}} \omega_y(O; \hat\eta)
      \Big(
         y \Big)
      - \mu^Y_{a}
   \Big\}
+ \mathbb{E}\!\Bigg(
   \frac{\mathbb{I}(A=a)}{\hat{\pi}_{a}}
     \big\{\mu^Y_{a} - \hat{\mu}^Y_{a}\big\}
   + \hat{\mu}^Y_{a} - \mu^Y_{a}
\Bigg) \Big\}\\[10pt]
&= \mathbb{E}\!\Bigg(
   \frac{\mathbb{I}(A=a)}{\hat{\pi}_{a}}
   \Big\{
     \sum_{y \in \mathcal{Y}_{A,C}} \omega_y(O; \hat\eta)
      \Big(
         y \Big)
      - \mu^Y_{a}
   \Big\}
+ \mathbb{E}\!\Bigg(
   \mathbb{E}[\frac{\mathbb{I}(A=a)}{\hat{\pi}_{a}}\mid C]
     \big\{\mu^Y_{a} - \hat{\mu}^Y_{a}\big\}
   + \hat{\mu}^Y_{a} - \mu^Y_{a}
\Bigg)\Big\}\\[10pt]
&= \mathbb{E}\!\Bigg(
   \frac{\mathbb{I}(A=a)}{\hat{\pi}_{a}}
   \Big\{
     \sum_{y \in \mathcal{Y}_{A,C}} \omega_y(O; \hat\eta)
      \Big(
         y \Big)
      - \mu^Y_{a}
   \Big\}
+ \mathbb{E}\!\Bigg(
   \frac{\pi_{a}}{\hat{\pi}_{a}}
     \big\{\mu^Y_{a} - \hat{\mu}^Y_{a}\big\}
   + \hat{\mu}^Y_{a} - \mu^Y_{a}
\Bigg)\Big\}\\[10pt]
&= \mathbb{E}\!\Bigg(
   \frac{\mathbb{I}(A=a)}{\hat{\pi}_{a}}
   \Big\{
      \sum_{y \in \mathcal{Y}_{A,C}} \omega_y(O; \hat\eta)
      \Big(
         y \Big)
      - \mu^Y_{a}
   \Big\}
+ \mathbb{E}\!\Bigg(
   \big\{\frac{\pi_{a}}{\hat{\pi}_{a}}-1\Big\}
     \big\{\mu^Y_{a} - \hat{\mu}^Y_{a}\big\}
\Bigg)\Big\}\\[10pt]
&= \mathbb{E}\!\Bigg(
   \frac{\mathbb{I}(A=a)}{\hat{\pi}_{a}}
   \Big\{
      \sum_{y \in \mathcal{Y}_{A,C}} \omega_y(O; \hat\eta)
      \Big(
         y \Big)
      - \mu^Y_{a}
   \Big\}
+ \mathbb{E}\!\Bigg(
   \big\{\frac{\pi_{a}-\hat{\pi}_{a}}{\hat{\pi}_{a}}\Big\}
     \big\{\mu^Y_{a} - \hat{\mu}^Y_{a}\big\}
\Bigg)\Big\}\\[10pt]
&\;\lesssim\; \mathbb{E}\!\Bigg(
   \frac{\mathbb{I}(A=a)}{\hat{\pi}_{a}}
   \Big\{
      \sum_{y \in \mathcal{Y}_{A,C}} \omega_y(O; \hat\eta)
      \Big(
         y \Big)
      - \mu^Y_{a}
   \Big\}\Bigg)
+
\|\pi_a - \hat{\pi}_a\| \;\|\mu_a^Y - \hat{\mu}_a^Y\|.\\[10pt]
\end{align*}

From Appendix~\ref{app:weights}, we have that

\begin{align*}
&\mathbb{E}\!\bigg[
\Big\{
\sum_{y_i \in \mathcal{Y}} \omega_{y_i}(O; \hat\eta)\, y_i
\Big\}
\;\Big|\; Y = y_j, A = a, C = c
\bigg] \\
&=
\sum_{y_i}
\Bigg\{
\mathbb{E}[h_{W,i}(W; \hat\eta)
\mid y_j,a,c]
\mathbb{E}[h_{Z, i}(
Z;\hat\eta) \mid y_j,a,c]\Bigg\} y_i\\[0.75em]
& \quad
+\sum_{y_i}
\Bigg\{
\mathbb{E}[h_{W,i}(W; \hat\eta)
\mid y_j,a,c]
\mathbb{E}[h_{V, i}(V;\hat\eta) \mid y_j,a,c]\Bigg\} y_i\\[0.75em]
& \quad
+\sum_{y_i}
\Bigg\{
\mathbb{E}[h_{Z,i}(Z; \hat\eta)
\mid y_j,a,c] 
\mathbb{E}[h_{V, i}(
V; \hat\eta) \mid y_j,a,c]\Bigg\} y_i\\[0.75em]
& \quad
-2 \sum_{y_i}\Bigg\{
\mathbb{E}[h_{W,i}(W; \hat\eta)
\mid y_j,a,c]
\mathbb{E}[h_{Z,i}(Z; \hat\eta)
\mid y_j,a,c]
\mathbb{E}[h_{V, i}(
V; \hat\eta) \mid y_j,a,c]\Bigg\} y_i \\
&=
\sum_{y_i}
\Bigg\{
\mathbb{E}[h_{W,i}(W; \hat\eta)
\mid y_j,a,c] - \mathbb{\hat E}[h_{W,i}(W; \hat\eta)
\mid y_j,a,c] + \mathbb{\hat E}[h_{W,i}(W; \hat\eta)
\mid y_j,a,c]\Bigg\}\\[0.75em]
& \qquad \cdot
\Bigg\{
\mathbb{E}[h_{Z, i}(
Z;\hat\eta) \mid y_j,a,c]
- \mathbb{\hat E}[h_{Z, i}(
Z;\hat\eta) \mid y_j,a,c]
+ \mathbb{\hat E}[h_{Z, i}(
Z;\hat\eta) \mid y_j,a,c]\Bigg\} y_i\\[0.75em]
& \quad
+ \sum_{y_i}
\Bigg\{
\mathbb{E}[h_{W,i}(W; \hat\eta)
\mid y_j,a,c] - \mathbb{\hat E}[h_{W,i}(W; \hat\eta)
\mid y_j,a,c] + \mathbb{\hat E}[h_{W,i}(W; \hat\eta)
\mid y_j,a,c]\Bigg\}\\[0.75em]
& \qquad \cdot
\Bigg\{
\mathbb{E}[h_{V, i}(
V;\hat\eta) \mid y_j,a,c]
- \mathbb{\hat E}[h_{V, i}(
V;\hat\eta) \mid y_j,a,c]
+ \mathbb{\hat E}[h_{V, i}(
V;\hat\eta) \mid y_j,a,c]\Bigg\} y_i\\[0.75em]
& \quad
+ \sum_{y_i}
\Bigg\{
\mathbb{E}[h_{Z,i}(Z; \hat\eta)
\mid y_j,a,c] - \mathbb{\hat E}[h_{Z,i}(Z; \hat\eta)
\mid y_j,a,c] + \mathbb{\hat E}[h_{Z,i}(Z; \hat\eta)
\mid y_j,a,c]\Bigg\}\\[0.75em]
& \qquad \cdot
\Bigg\{
\mathbb{E}[h_{V, i}(
V;\hat\eta) \mid y_j,a,c]
- \mathbb{\hat E}[h_{V, i}(
V;\hat\eta) \mid y_j,a,c]
+ \mathbb{\hat E}[h_{V, i}(
V;\hat\eta) \mid y_j,a,c]\Bigg\} y_i\\[0.75em]
& \quad
-2 \sum_{y_i}
\Bigg\{
\mathbb{E}[h_{W,i}(W; \hat\eta)
\mid y_j,a,c] - \mathbb{\hat E}[h_{W,i}(W; \hat\eta)
\mid y_j,a,c] + \mathbb{\hat E}[h_{W,i}(W; \hat\eta)
\mid y_j,a,c]\Bigg\}\\[0.75em]
& \qquad \cdot
\Bigg\{
\mathbb{E}[h_{Z, i}(
Z;\hat\eta) \mid y_j,a,c]
- \mathbb{\hat E}[h_{Z, i}(
Z;\hat\eta) \mid y_j,a,c]
+ \mathbb{\hat E}[h_{Z, i}(
Z;\hat\eta) \mid y_j,a,c]\Bigg\}\\[0.75em]
& \qquad \cdot
\Bigg\{
\mathbb{E}[h_{V, i}(
V;\hat\eta) \mid y_j,a,c]
- \mathbb{\hat E}[h_{V, i}(
V;\hat\eta) \mid y_j,a,c]
+ \mathbb{\hat E}[h_{V, i}(
V;\hat\eta) \mid y_j,a,c]\Bigg\} y_i,
\end{align*}
where
\[
\mathbb{\hat E}[h_{W}(W; \hat\eta)\mid Y=y_j,a,c]
=
\mathbf{\hat M}_W^{-1}(a,c)\mathbf{\hat M}_W(a,c) e_j,
\]
\[
\mathbb{\hat E}[h_{Z}(Z; \hat\eta)\mid Y=y_j,a,c]
=
\mathbf{\hat M}_Z^{-1}(a,c)\mathbf{\hat M}_Z(a,c) e_j,
\]
\[
\text{ and }
\mathbb{\hat E}[h_{V}(V; \hat\eta)\mid Y=y_j,a,c]
=
\mathbf{\hat M}_V^{-1}(a,c)\mathbf{\hat M}_V(a,c) e_j,
\]
so that
\[
\mathbb{\hat E}[h_{W,i}(W; \hat\eta)\mid Y=y_j,a,c]
=
\mathbf 1\{i=j\},
\]
\[
\mathbb {\hat E}[h_{Z,i}(Z; \hat \eta)\mid Y=y_j,a,c]
=
\mathbf 1\{i=j\},
\]
\[ \text{ and }
\mathbb {\hat E}[h_{V,i}(V; \hat \eta)\mid Y=y_j,a,c]
=
\mathbf 1\{i=j\}.
\]

Then by conditions (iii)--(v),
\begin{align*}
&\mathbb{E}\!\bigg[
\Big\{
\sum_{y_i \in \mathcal{Y}} \omega_{y_i}(O; \hat\eta)\, y_i
\Big\}
\;\Big|\; Y = y_j, A = a, C = c
\bigg] \\[0.75em]
&=  
\Bigg\{
\mathbb{E}[h_{W,j}(W; \hat\eta)
\mid y_j,a,c] - \mathbb{\hat E}[h_{W,j}(W; \hat\eta)
\mid y_j,a,c] + 1\Bigg\}\\[0.75em]
& \qquad \cdot
\Bigg\{
\mathbb{E}[h_{Z, j}(
Z;\hat\eta) \mid y_j,a,c]
- \mathbb{\hat E}[h_{Z, j}(
Z;\hat\eta) \mid y_j,a,c]
+ 1\Bigg\} y_j\\[0.75em]
& 
+ 
\Bigg\{
\mathbb{E}[h_{W,j}(W; \hat\eta)
\mid y_j,a,c] - \mathbb{\hat E}[h_{W,j}(W; \hat\eta)
\mid y_j,a,c] + 1\Bigg\}\\[0.75em]
& \qquad \cdot
\Bigg\{
\mathbb{E}[h_{V, j}(
V;\hat\eta) \mid y_j,a,c]
- \mathbb{\hat E}[h_{V, j}(
V;\hat\eta) \mid y_j,a,c]
+ 1\Bigg\} y_j\\[0.75em]
& 
+ 
\Bigg\{
\mathbb{E}[h_{Z,j}(Z; \hat\eta)
\mid y_j,a,c] - \mathbb{\hat E}[h_{Z,j}(Z; \hat\eta)
\mid y_j,a,c] + 1\Bigg\}\\[0.75em]
& \qquad \cdot
\Bigg\{
\mathbb{E}[h_{V, j}(
V;\hat\eta) \mid y_j,a,c]
- \mathbb{\hat E}[h_{V, j}(
V;\hat\eta) \mid y_j,a,c]
+ 1\Bigg\} y_j\\[0.75em]
&
-2 
\Bigg\{
\mathbb{E}[h_{W,j}(W; \hat\eta)
\mid y_j,a,c] - \mathbb{\hat E}[h_{W,j}(W; \hat\eta)
\mid y_j,a,c] + 1\Bigg\}\\[0.75em]
& \qquad \cdot
\Bigg\{
\mathbb{E}[h_{Z, j}(
Z;\hat\eta) \mid y_j,a,c]
- \mathbb{\hat E}[h_{Z, j}(
Z;\hat\eta) \mid y_j,a,c]
+ 1\Bigg\}\\[0.75em]
& \qquad \cdot
\Bigg\{
\mathbb{E}[h_{V, j}(
V;\hat\eta) \mid y_j,a,c]
- \mathbb{\hat E}[h_{V, j}(
V;\hat\eta) \mid y_j,a,c]
+ 1\Bigg\} y_j + o_p(n^{-1/2}).\\
\end{align*}

Applying conditions (iii)--(v) again after canceling terms, we have \[\mathbb{E}\!\bigg[
\Big\{
\sum_{y_i \in \mathcal{Y}} \omega_{y_i}(O; \hat\eta)\, y_i
\Big\}
\;\Big|\; Y = y_j, A = a, C = c
\bigg] = y_j + o_p(n^{-1/2}).\]

Then,
\begin{align*}
R_2 
&\lesssim \mathbb{E}\!\Bigg(
   \frac{\mathbb{I}(A=a)}{\hat{\pi}_{a}}
   \Bigg\{
      \sum_{y \in \mathcal{Y}_{A,C}} \omega_{y}(O; \hat\eta)\, y
      - \mu^Y_{a}
   \Bigg\}
\Bigg) + \lVert\pi_a - \hat{\pi}_a\rVert\;
          \big\|\,\mu_a^Y - \hat{\mu}_a^Y\rVert \\[6pt]
&= \mathbb{E}\!\Bigg(
   \frac{\mathbb{I}(A=a)}{\hat{\pi}_{a}}
   \mathbb{E}\left[\Bigg\{
      \sum_{y \in \mathcal{Y}_{A,C}} \omega_{y}(O; \hat\eta)\, y
      - \mu^Y_{a}
   \Bigg\}\mid Y, A, C\right]
\Bigg) + \lVert\pi_a - \hat{\pi}_a\rVert\;
          \big\|\,\mu_a^Y - \hat{\mu}_a^Y\rVert \\
&=\mathbb{E}\!\Bigg(
   \frac{\mathbb{I}(A=a)}{\hat{\pi}_{a}}
   \{Y - \mu^Y_{a}\}
\Bigg) + \lVert\pi_a - \hat{\pi}_a\rVert\;
          \big\|\,\mu_a^Y - \hat{\mu}_a^Y\rVert + o_p(n^{-1/2})  \\
&= 0 + \lVert\pi_a - \hat{\pi}_a\rVert\;
          \big\|\,\mu_a^Y - \hat{\mu}_a^Y\rVert + o_p(n^{-1/2}). 
\end{align*}

Finally, by (ii), we have that $R_2=o_p(n^{-1/2}).$
\end{proof}

\section{Binary Data-Generating Process}\label{app:binary_dgp}
Confounder \(C\) and treatment \(A\) are generated as follows:
\[
C \sim \mathrm{Bernoulli}(0.3), \quad 
A \mid C \sim 
\begin{cases}
\mathrm{Bernoulli}(0.4), & C=0, \\
\mathrm{Bernoulli}(0.6), & C=1.
\end{cases}
\]
Additionally, outcomes are generated according to a logistic (expit) model:
\[
Y \mid A, C \sim \mathrm{Bernoulli}\!\big(\mathrm{expit}(-1 + 3A - 1.5C)\big).
\]  
Finally, proxies \(Z, W, V\) are independently generated conditional on \(Y\) as follows:
\[
p(Z=1 \mid Y=1) = p(W=1 \mid Y=1) = p(V=1 \mid Y=1) = 0.9,
\]  
\[
p(Z=1 \mid Y=0) = p(W=1 \mid Y=0) = p(V=1 \mid Y=0) = 0.1.
\]  

By construction,
\[
W \perp Z \perp V \mid Y.
\]

We form the empirical three-way array/tensor
\[
\widehat{\mathcal T}(w,z,v)
=
\widehat{p}(W=w,Z=z,V=v).
\]
Under Assumption~\ref{assump:ci_hidden_y},
\[
p(W,Z,V)
=
\sum_{y=0}^1
p(Y=y)
p(W \mid Y=y)
p(Z \mid Y=y)
p(V \mid Y=y).
\]
Eigendecomposition tasks applied to \(\widehat{\mathcal T}(w,z,v)\) discussed in Section~\ref{sec:related_work} exploit this structure to yield estimates of components
\[
\widehat{p}(Y),
\qquad
\widehat{p}(W \mid Y),
\qquad
\widehat{p}(Z \mid Y),
\qquad
\widehat{p}(V \mid Y),
\]
up to labeling of \(Y\).
Note that although these nuisance functions may alternatively be recovered using other estimation strategies, we employ eigendecomposition tasks to avoid iterative optimization procedures and improve computational efficiency. 

Labels are assigned by ordering the recovered latent components according to \(\widehat{p}(W=1 \mid Y)\), since it is strictly increasing in \(Y\) under the data-generating process (option 2 in Assumption~\ref{assump:unique_mapping}). 

The recovered measurement distribution \(\widehat p(W \mid Y)\) is then used together with the observed distribution \(\widehat p(A,W \mid C)\) to recover \(\widehat p(A,Y \mid C)\). Specifically, 
\[
p(A,W \mid C)
=
\sum_y
p(A,Y=y \mid C)\,
p(W \mid Y=y).
\]
Solving this system yields estimate \(\widehat p(A,Y \mid C)\), which determines the outcome regression estimate \(\widehat p(Y \mid A,C)\). 

The propensity score estimate, \(\widehat p(A \mid C)\), is constructed directly from the observed data.  

\section{Mixed DGP 1 (Binary Outcome; Continuous Confounders and Proxies)}\label{app:mixed_dgp1}

Confounder \(C\), treatment \(A\), and binary outcome \(Y\) are generated according to
\[
C \sim \mathcal{N}(0,1);
\]
\[
A \mid C \sim \mathrm{Bernoulli}\!\left(
\operatorname{expit}(-0.3 + 0.8C)
\right);
\]
\[
Y \mid A,C \sim \mathrm{Bernoulli}\!\left(
\operatorname{expit}(\beta_0 + \beta_A A + \beta_C C)
\right),
\]
with \((\beta_0,\beta_A,\beta_C)=(-3,3,0.5)\). The continuous proxies \(W\), \(Z\), and \(V\) are generated independently conditional on \(Y\) according to
\[
W \mid Y=y \sim \mathcal N(-0.3 + 1.5y,\,1),
\]
\[
Z \mid Y=y \sim \mathcal N(0.5 + 1.2y,\,1),
\]
\[
V \mid Y=y \sim \mathcal N(1.0 + 0.8y,\,1).
\]

Thus,
\[
W \perp Z \perp V \mid Y.
\]
We directly model the continuous measurement distribution using a two-component latent Gaussian mixture model. Under Assumption~\ref{assump:ci_hidden_y},
\[
p(W,Z,V)
=
\sum_{y=0}^1
p(Y=y)\,
p(W \mid Y=y)\,
p(Z \mid Y=y)\,
p(V \mid Y=y).
\]

We use Gaussian measurement models
\[
W \mid Y=y \sim \mathcal{N}(\mu_{W,y},\sigma_{W,y}^2),
\]
with analogous models for \(Z\) and \(V\) to estimate parameters
\[
p(Y=y),
\qquad
(\mu_{W,y},\sigma_{W,y}),
\qquad
(\mu_{Z,y},\sigma_{Z,y}),
\qquad
(\mu_{V,y},\sigma_{V,y})
\]
via an EM algorithm maximizing the observed-data log-likelihood
\[
\ell
=
\sum_{i=1}^n
\log
\left[
\sum_{y=0}^1
p(Y_i=y)\,
p(W_i \mid Y_i=y)\,
p(Z_i \mid Y_i=y)\,
p(V_i \mid Y_i=y)
\right].
\]

Since the latent classes are only identified up to label switching, latent outcome labels are determined using the recovered Gaussian measurement distributions \(\{\widehat p(W \mid Y=y): y \in \{0,1\}\}\) via their estimated means \(\widehat\mu_{W,y}\). Specifically, because the conditional mean of \(W\) is strictly increasing in \(Y\) under the data-generating process, satisfying option 2 in Assumption~\ref{assump:unique_mapping}, the recovered distribution with larger estimated mean is identified with \(Y=1\).
Nuisance functions in \(\omega(O;\eta)\) are estimated via the labeled quantities \(\widehat{\mu}_{W,y}\), \(\widehat{\mu}_{Z,y}\), and \(\widehat{\mu}_{V,y}\).

Next, we estimate the outcome regression model
\[
p(Y=1 \mid A,C)
=
\operatorname{expit}(\beta_0+\beta_AA+\beta_CC).
\]
The previously recovered parameters determine the estimated measurement distributions \(\widehat p(W \mid y)\), \(\widehat p(Z \mid y)\), and \(\widehat p(V \mid y)\). We then estimate the outcome regression model \(p(Y=y \mid A,C)\) by maximizing the observed-data log-likelihood
\[
\ell(\beta)
=
\sum_{i=1}^n
\log
\left[
\sum_{y=0}^1
\Pr_\beta(Y_i=y \mid A_i,C_i)
\widehat p(W_i \mid y)\,
\widehat p(Z_i \mid y)\,
\widehat p(V_i \mid y)
\right].
\]

Thus, we first recover the latent Gaussian measurement model \(p(W,Z,V \mid Y)\), which is then used to estimate the hidden outcome regression \(p(Y \mid A,C)\) through the induced observed-data likelihood. Although nuisance functions can be jointly estimated using a single expectation-maximization procedure, we instead use this two-step procedure due to its computational efficiency under this DGP. 

The propensity score \(p(A \mid C)\) is estimated directly from the observed data using a logistic regression model.

\section{Mixed DGP 2 (Three-category Outcome; Continuous Confounders and Proxies)}\label{app:mixed_dgp2}

Confounder \(C\), treatment \(A\), and three-category outcome \(Y\) are generated according to

\[
C \sim \mathcal{N}(0,1);
\]

\[
A \mid C \sim \mathrm{Bernoulli}\!\left(
\operatorname{expit}(-0.3 + 0.8C)
\right);
\]

\[
Y \mid A,C \sim \mathrm{Multinomial}\!\left(
1;
\pi_1(A,C),\pi_2(A,C),\pi_3(A,C)
\right), \qquad Y \in \{0,1,2\},
\]

where

\[
\pi_2(A,C)
=
\frac{
\exp(\beta_{20} + \beta_{2A} A + \beta_{2C} C)
}{
1
+
\exp(\beta_{20} + \beta_{2A} A + \beta_{2C} C)
+
\exp(\beta_{30} + \beta_{3A} A + \beta_{3C} C)
},
\]

\[
\pi_3(A,C)
=
\frac{
\exp(\beta_{30} + \beta_{3A} A + \beta_{3C} C)
}{
1
+
\exp(\beta_{20} + \beta_{2A} A + \beta_{2C} C)
+
\exp(\beta_{30} + \beta_{3A} A + \beta_{3C} C)
},
\]

and

\[
\pi_1(A,C)
=
\frac{
1
}{
1
+
\exp(\beta_{20} + \beta_{2A} A + \beta_{2C} C)
+
\exp(\beta_{30} + \beta_{3A} A + \beta_{3C} C)
}.
\]

We take

\[
(\beta_{20},\beta_{2A},\beta_{2C})=(-1.5,1.2,0.4),
\]

and

\[
(\beta_{30},\beta_{3A},\beta_{3C})=(-3.0,2.4,0.7).
\]

After sampling the latent outcome \(Y\), the continuous proxies \(W\), \(Z\), and \(V\) are generated independently conditional on \(Y\) according to
\[
W \mid Y=y \sim \mathcal N(\mu_{W,y},\sigma_{W,y}^2),
\qquad
Z \mid Y=y \sim \mathcal N(\mu_{Z,y},\sigma_{Z,y}^2),
\qquad
V \mid Y=y \sim \mathcal N(\mu_{V,y},\sigma_{V,y}^2),
\]
with class-specific means
\[
(\mu_{W,1},\mu_{W,2},\mu_{W,3})=(-2.5,\,2.0,\,6.0),
\]
\[
(\mu_{Z,1},\mu_{Z,2},\mu_{Z,3})=(-1.5,\,3.0,\,6.5),
\]
\[
(\mu_{V,1},\mu_{V,2},\mu_{V,3})=(-3.0,\,2.0,\,5.8),
\]
and class-specific standard deviations
\[
(\sigma_{W,1},\sigma_{W,2},\sigma_{W,3})=(0.95, 1.25, 1.35),
\]
\[
(\sigma_{Z,1},\sigma_{Z,2},\sigma_{Z,3})=(0.95, 1.10, 1.20),
\]
\[
(\sigma_{V,1},\sigma_{V,2},\sigma_{V,3})=(0.90, 1.05, 1.15).
\]

Thus,
\[
W \perp Z \perp V \mid Y.
\]
We directly model the continuous measurement distribution using a three-component latent Gaussian mixture model. Under Assumption~\ref{assump:ci_hidden_y},
\[
p(W,Z,V)
=
\sum_{y=1}^3
p(Y=y)\,
p(W \mid Y=y)\,
p(Z \mid Y=y)\,
p(V \mid Y=y).
\]

The latent Gaussian measurement parameters
\[
p(Y=y),
\qquad
(\mu_{W,y},\sigma_{W,y}),
\qquad
(\mu_{Z,y},\sigma_{Z,y}),
\qquad
(\mu_{V,y},\sigma_{V,y})
\]
are estimated via an EM algorithm maximizing the observed-data log-likelihood

\[
\ell
=
\sum_{i=1}^n
\log
\left[
\sum_{y=1}^3
p(Y_i=y)\,
p(W_i \mid Y_i=y)\,
p(Z_i \mid Y_i=y)\,
p(V_i \mid Y_i=y)
\right].
\]

Since the latent classes are only identified up to label switching, latent outcome labels are determined using the recovered Gaussian measurement distributions \(\{\widehat p(W \mid Y=y): y \in \{0,1,2\}\}\) via their estimated means \(\widehat\mu_{W,y}\). Specifically, because the conditional mean of \(W\) is strictly increasing in \(Y\) under the data-generating process, satisfying option 2 in Assumption~\ref{assump:unique_mapping}, the recovered distribution with the smallest estimated mean is identified with \(Y=1\), the intermediate estimated mean with \(Y=2\), and the largest estimated mean with \(Y=3\). Nuisance functions in \(\omega(O;\eta)\) are estimated via the labeled quantities \(\widehat{\mu}_{W,y}\), \(\widehat{\mu}_{Z,y}\), \(\widehat{\mu}_{V,y}\), \(\widehat{\sigma}_{W,y}\), \(\widehat{\sigma}_{Z,y}\), and \(\widehat{\sigma}_{V,y}\).

Next, we estimate the outcome regression model
\[
p(Y=y \mid A,C),
\qquad y\in\{0,1,2\}.
\]

The previously recovered parameters determine the estimated measurement distributions \(\widehat p(W \mid y)\), \(\widehat p(Z \mid y)\), and \(\widehat p(V \mid y)\). We then estimate the outcome regression model \(p(Y=y \mid A,C)\) by maximizing the observed-data log-likelihood
\[
\ell(\beta)
=
\sum_{i=1}^n
\log
\left[
\sum_{y=1}^3
\Pr_\beta(Y_i=y \mid A_i,C_i)
\widehat p(W_i \mid y)\,
\widehat p(Z_i \mid y)\,
\widehat p(V_i \mid y)
\right].
\]

We then use the estimated parameters \(\widehat{\beta}\) to obtain
\(\widehat{\mathbb{E}}[Y \mid A,U]\).

The propensity score \(p(A \mid C)\) is estimated directly from the observed data using a logistic regression model.

\end{document}